\documentclass{article}
\usepackage{amsfonts}
\usepackage{amsmath,amsthm}
\usepackage{hyperref}

\usepackage[all]{xy}
\usepackage{graphicx}

\usepackage{amssymb}
\usepackage{latexsym}

\DeclareMathAlphabet\EuFrak{U}{euf}{m}{n}	
\SetMathAlphabet\EuFrak{bold}{U}{euf}{b}{n}	

\newcommand{\lra}{\leftrightarrow}

\newcommand{\da}{\downarrow}
\newcommand{\hra}{\hookrightarrow}

\newcommand{\ovl}{\overline}

\newcommand{\wa}{\widehat}
\newcommand{\wt}{\widetilde}

\newcommand{\sC}{{\it C*}-}

\newcommand{\bC} {{\mathbb C}}

\newcommand{\bR} {{\mathbb R}}
\newcommand{\bT} {{\mathbb T}}

\newcommand{\bSU} {{\mathbb{SU}}}

\newcommand{\bZ} {{\mathbb Z}}

\newcommand{\bN} {{\mathbb N}}
\newcommand{\bQ} {{\mathbb Q}}

\newcommand{\bK} {{\mathbb K}}

\newcommand{\ud}{{{\mathbb U}(d)}}

\newcommand{\eps}{\epsilon}
\newcommand{\vareps}{\xi}

\newcommand{\mD}{\mathcal D}

\newcommand{\mG}{\mathcal G}

\newcommand{\mO}{\mathcal O}

\newcommand{\mS}{\mathcal S}

\newcommand{\mV}{\mathcal V}

\newcommand{\efA}{\EuFrak{A}}
\newcommand{\efB}{\EuFrak{B}}
\newcommand{\efC}{\EuFrak{C}}

\newcommand{\efF}{\EuFrak{F}}
\newcommand{\efG}{\EuFrak{G}}
\newcommand{\efH}{\EuFrak{H}}

\newcommand{\efP}{\EuFrak{P}}

\newcommand{\efR}{\EuFrak{R}}
\newcommand{\efS}{\EuFrak{S}}
\newcommand{\efT}{\EuFrak{T}}

\newcommand{\efZ}{\EuFrak{Z}}

\newcommand{\bo}{{\partial_0 b}}
\newcommand{\bl}{{\partial_1 b}}

\newcommand{\ad}{{\mathrm{ad}}}

\newcommand{\obj}{{\bf obj \ }}

\newtheorem{thm}{Theorem}[section]
\newtheorem{cor}[thm]{Corollary}
\newtheorem{lem}[thm]{Lemma}

\newtheorem{defn}[thm]{Definition}

\newtheorem{rem}[thm]{Remark}

\theoremstyle{definition}
\newtheorem{ex}{Example}[section]

\theoremstyle{remark}

\numberwithin{equation}{section}

\begin{document}

\author{{\sf Ezio Vasselli}
                         \\{\sf ezio.vasselli@gmail.com}}

\title{Presheaves of superselection structures in curved spacetimes}
\maketitle

\begin{abstract}
We show that superselection structures on curved spacetimes, 
that are expected to describe quantum charges affected by the underlying geometry, 
are categories of sections of presheaves of symmetric tensor categories.
When an embedding functor is given,
the superselection structure is a Tannaka-type dual of a locally constant group bundle, 
which hence becomes a natural candidate for the role of gauge group.
Indeed, we show that any locally constant group bundle (with suitable structure group)
acts on a net of \sC algebras fulfilling normal commutation relations
on an arbitrary spacetime.
We also give examples of gerbes of \sC algebras, defined by Wightman fields and
constructed using projective representations of the fundamental group of the spacetime,
that we propose as solutions for the problem that existence and uniqueness of the 
embedding functor are not guaranteed.
\end{abstract}

\tableofcontents
\markboth{Contents}{Contents}

\section{Introduction.}
\label{intro}

In the ordinary approach to the theory of particle physics quantum fields are fundamental objects, 
but do not necessarily yield quantum observables because the property of gauge invariance 
with respect to a symmetry group must be imposed.

In algebraic quantum field theory this point of view is subverted:
at a first stage only quantum observables are considered, 
in the spirit that these contain all the physical informations of the system under consideration.
Thus one could ask whether the observable quantities determine the quantum fields and the gauge group,
and the answer to this question is a pivotal result:
quantum fields and gauge group can be reconstructed, 
starting from the observable algebra and the set, called the superselection structure, of its physically relevant Hilbert space representations
({\em sectors}), that are interpreted as quantum charges (\cite{DR90}).
A celebrated mathematical byproduct is the Doplicher-Roberts duality for compact groups vs. 
{\em abstract} symmetric tensor categories, that generalizes Tannaka-Krein duality (\cite{DR89}).

\smallskip

The mathematical scenario of the above construction is the one defined by Haag, Kastler and then Araki 
in the middle sixties (\cite[Chp.4]{Ara},\cite[Chp.III]{Haa}): 
we have a base $\Delta_{dc}$ of the Minkowski spacetime, the one of double cones, and, for any $o \in \Delta_{dc}$, a \sC algebra $A_o$ that is interpreted as the one of quantum observables localized within $o$. 
Starting from this interpretation, it is natural to require that there are inclusion *-morphisms
\begin{equation}
\label{eq.00}
\jmath_{o'o} : A_o \to A_{o'}
\ \ : \ \
\jmath_{o'' o'} \circ \jmath_{o'o} = \jmath_{o'' o}
\ \ , \ \
\forall o \subseteq o' \subseteq o'' \in \Delta_{dc}
\ .
\end{equation}
The pair $\efA := (A,\jmath)$ is known as the {\em observable net} and is usually -- but not always -- realized by means
of quantum fields in the sense of Wightman (\cite[\S 4.8]{Ara}). 
An important point is that $\Delta_{dc}$ is a directed {\em poset} (partially ordered set) when ordered under inclusion, 
and this implies that we can construct the inductive limit \sC algebra $\vec{A} := \lim_j (A,\jmath)$
that can be represented as a *-algebra of operators on a fixed Hilbert space. 
The superselection structure is then realized as a symmetric tensor \sC category $T$ of *-endomorphisms of $\vec{A}$, 
and characterized as the dual of the compact gauge group $G$.
This means, in particular, that any $\rho \in \obj T$ defines a finite-dimensional, unitary representation 
\begin{equation}
\label{eq.rep}
u_\rho : G \to UH_\rho
\end{equation}
such that $H_\rho$ appears in the field algebra as the $G$-vector space spanned by a multiplet of field operators 
$\psi_1 , \ldots , \psi_{{\mathrm{dim}}(u_\rho)}$,
see (\cite{DR90}).

\smallskip

Now, for generic spacetimes $M$ in algebraic quantum field theory it is customary to select a convenient base $\Delta$, 
according to the causal properties and symmetries of $M$.
In in this way one obtains, for example, 
the base of intervals with proper closure for the circle in conformal theory,
and the base of \emph{diamonds} generating the topology of a globally hyperbolic spacetime, 
that is the setting of general relativity
{\footnote{
In principle we may find more that a suitable base for the purposes of algebraic quantum field theory;
but the superselection structure turns out to be independent of the base up to isomorphism, see \cite[Theorem 2.23]{Ruz05}.
}}.
We note that when elements of $\Delta$ are arcwise and simply connected (as intervals, double cones and diamonds),
we are able to recover base-independent invariants of $M$:
the fundamental group $\pi_1(M)$ (\cite{Ruz05}), 
the first singular and de Rham cohomology (\cite{RRV08})
and the category of flat bundles (\cite{RRV07,RRV08}).
In several cases of physical interest, in particular when $\pi_1(M)$ is not trivial
(anti-de Sitter spacetimes, the circle, Aharonov-Bohm solenoids),
$\Delta$ cannot be directed, see (\ref{eq.back1}), and $\efA$ is more generally a \sC precosheaf.

\smallskip

At the mathematical level, at the present time the most general superselection structure 
available on spacetimes with at least two spatial dimensions is the one defined by Brunetti and Ruzzi (\cite{BR08}):
it is expressed in terms of cocycles $z$ with values in the unitary precosheaf of $\efA$,
that play the role of \emph{charge transporters} in the terminology of \cite{Rob0,Rob}.
This generalizes previous works by Guido, Longo, Roberts, Verch and Ruzzi (\cite{Rob0,Rob,GLRV01,Ruz05}),
with the motivation of clarifying the relations between spacetime topology and structural properties of quantum field theory.
The difference between the approach of (\cite{BR08}) and the one of \cite{GLRV01,Ruz05} is that in the first case we have
that observable unitary phases (that we call \emph{holonomies}) appear when one performs a charge transport around a closed loop
(see \cite[\S 7]{BR08}).
This feature is typical in situations like the Aharonov-Bohm effect (\cite[\S 15.5]{Fey},\cite{ES49,AB59}), 
where the electromagnetic field gives rise to observable effects induced by the holonomy of the vector potential.
The sectors studied in \cite{GLRV01,Ruz05}, instead, have trivial holonomy in the above sense
and for this reason, following \cite{BR08}, we call them \emph{topologically trivial sectors}.

\smallskip

The Brunetti-Ruzzi superselection structure yields a symmetric tensor \sC category $Z^{1,\bullet}(\efA)$ describing 
charge composition, statistics and the particle-antiparticle correspondence (\emph{conjugate sectors}), 
contains the set $Z^{1,\bullet}_t(\efA)$ of topologically trivial sectors as a full subcategory, and coincides with it when $\pi_1(M)$ is trivial.
Nevertheless at the present time it is not clear how to perform a field reconstruction as in the cases of 
Minkowski spacetime (\cite{DR90}) and topologically trivial sectors in spacetimes where $\Delta$ is directed (\cite{GLRV01}), 
so we do not know whether it is possible to interpret the dual group of $Z^{1,\bullet}(\efA)$ as the gauge group.
We stress that this problem is open also for topologically trivial sectors in spacetimes with nontrivial fundamental group.

\

In the present paper we focus on globally hyperbolic spacetimes $M$ with at least two spatial dimensions.
This last condition is motivated by the fact that we are interested in sectors having ordinary Fermi-Bose (para)statistics
that, as well-known, in low-dimensional spacetimes are not ensured (braided statistics may appear, see \cite[\S 3]{GLRV01}).
Following the standard approach (\cite{GLRV01,Ruz05,BR08}) we take $\Delta$ as the base of diamonds
whose basic properties, related to the features of Cauchy surfaces in $M$, are discussed in \cite[\S 2.1]{GLRV01} and \cite[\S 3.2]{Ruz05}.

\smallskip

In one of our main results we prove that $Z^{1,\bullet}(\efA)$ is the category of sections of 
a presheaf $\efS^\bullet$ of symmetric tensor \sC categories with simple units (Theorem \ref{thm.br}).
Our intention is to use $\efS^\bullet$ as a tool to perform the field reconstruction 
(see the following paragraphs and \S \ref{sec.concl}), 
anyway our result also gives new light on the structure of $Z^{1,\bullet}(\efA)$: 
first, it allows us to assign to sectors new invariants of geometric nature, 
and, secondly, it suggests that the dual object describing the superselection structure may not to be merely a group.

In fact, as a consequence of abstract results for generic presheaves (\cite{Vas12}), 
we have that $Z^{1,\bullet}(\efA)$ can be characterized as a Tannaka-type dual $\wa \mG$ of 
a locally constant -- or, equivalently, flat -- group bundle 
\[
\mG \to M \ .
\]
%
Thus we have $\mG$-equivariant flat vector bundles in place of the unitary representations (\ref{eq.rep}), and, 
keeping in mind the above-mentioned interpretation of $G$-Hilbert spaces in terms of multiplets of fields, 
this suggests that the charged field operators that are expected to give rise to sectors in $Z^{1,\bullet}(\efA)$
should be visualized in terms of local sections of flat bundles. 
This idea is not new and has been formulated in the setting of quantum fields (\cite{Ish,AI79}), 
and is currently object of research (\cite{FL}).

An essential ingredient of the duality result is given by an embedding morphism 
$I : \efS^\bullet \to \efC$,
where $\efC$ is a presheaf of full subcategories of the one of Hilbert spaces.
This is a crucial point: $I$ may not exist or not to be unique,
so the same holds for $\mG$ which, when existing, will depend on the choice of $I$.
This does not exclude that there is a criterion to select a unique "physical" gauge group bundle $\mG_{phys}$, 
nevertheless our results show that the isomorphism
\begin{equation}
\label{eq.01}
Z^{1,\bullet}(\efA) \ \simeq \ \wa \mG_{phys} 
\end{equation}
does not suffice to characterize $\mG_{phys}$, differently from what happens in the Minkowski spacetime.

From the point of view of invariants, we prove that any sector $z$ defines a compact Lie group 
$G_z \subseteq \ud$, $d \in \bN$, 
and a projective representation 
\[
\chi_z : \pi_1(M) \to NG_z / G_z \ ,
\]
where $NG_z$ is the normalizer of $G_z$ in $\ud$. 
The importance of $\chi_z$ relies in the fact that an embedding of $\efS^\bullet$ defines, for any sector $z$, a lift
\[
\wa \chi_z : \pi_1(M) \to NG_z
\ \ , \ \
\chi_z = \wa \chi_z \, {\mathrm{mod}} G_z \ .
\]
Lifts of $\chi_z$ are classified in cohomological terms by a $NG_z$-valued 1--cochain $u_z$ 
defined on the simplicial set of $\Delta$ (see Theorem \ref{thm.gerbe}).
We interpret $\chi_z$ as a flat principal $NG_z / G_z$-bundle
and assign to it Cheeger-Chern-Simons characteristic classes living in the odd cohomology of $M$ (see (\ref{eq.ccs}))
that vanish when $z \in Z^{1,\bullet}_t(\efA)$.

\smallskip

To illustrate how group bundles yield gauge symmetries we construct field \sC precosheaves 
obtained by twisting a given "topologically trivial" field net with gauge group $G$ (Lemma \ref{lem.gauge1}).
The resulting \sC precosheaf is acted upon by a locally constant $G$-bundle and has representations with nontrivial holonomy in the sense of \cite{BR08},
with the feature that non-triviality of the holonomy is explicitly related with non-triviality of the gauge group bundle.
In accordance with the above-mentioned duality breaking, we have examples of inequivalent field precosheaves 
-- acted upon by inequivalent group bundles --
having the same observable net, this last represented on a fixed Hilbert space (Cor.\ref{cor.gauge1}).
In \S \ref{rem.qed} we discuss twists given by background classical potentials interacting with Dirac fields,
obtaining relative phases analogous to those appearing in the Aharonov-Bohm effect.
In \S \ref{sec.free},
by using generalized free fields we show that for every locally constant $G$-bundle (with suitable structure group) 
there is a field \sC precosheaf having it as a "group" of global gauge symmetries.
The choice of considering free fields is dictated by the current lack (at least at the author knowledge) of 
nonperturbative, local quantum fields in spacetimes with nontrivial $\pi_1(M)$ and spacetime dimension $\geq 3$;
but since our twisting procedure works independently of the spacetime dimension, 
in \S \ref{sec.BMT} we apply our result to conformal current algebras, where the cases $G = \bZ_{2n}$, $n \in \bN$, occur (\cite{BMT1,BMT2}).

\smallskip

The loss of uniqueness and -- in theory -- of existence of the dual group bundle can be avoided already at the mathematical level, 
moving ourselves  in the more convenient setting of gerbes over $\Delta$. 
This notion is illustrated in \S \ref{sec.concl}.
We shall show in a forthcoming paper (\cite{VasX}) that $\efS^\bullet$ defines a \emph{unique} dual group gerbe $\check \efG$, 
which {\em may} collapse to several, possibly inequivalent, group bundles having the same dual.
In this spirit, we prove in (\ref{eq.CON1a}) that the above-mentioned cochains $u_z$ define Lie group gerbes, 
that will turn out to be finite-dimensional representations of $\check \efG$.

On these grounds we expect that, at least \emph{a priori}, existence and uniqueness of
a field \sC precosheaf having $\efA$ as fixed-point precosheaf are not ensured.
Thus we give the notion of {\em $\check \efG$-\sC gerbe} over a poset, defined as a pair $\efF = (F,i)$ where 
$F := \{ F_o \}_{o \in \Delta}$ 
is a family of \sC algebras with *-monomorphisms 
$i_{o'o} : F_o \to F_{o'}$, $o \subseteq o'$, 
whose possible precosheaf structures are encoded by $\check \efG$ 
(see (\ref{eq.CON2}): an interesting point is that fixed-point algebras of a $\check \efG$-\sC gerbe always have a unique precosheaf structure).
The fact that a unique "field gerbe" $\efF$ may be reconstructed starting from $\efA$ and $Z^{1,\bullet}(\efA)$ is object of a work in progress, 
in which the presheaf $\efS^\bullet$ plays a crucial role (see \S \ref{sec.concl}).
In this setting, we expect that some invariant of $Z^{1,\bullet}(\efA)$, beside (\ref{eq.01}),
should allow us to pick a unique "physical" field \sC precosheaf in which $\efF$ would collapse, carrying a $\mG_{phys}$-action.
Constructions of $\check \efG$-\sC gerbes, 
defined by Wightman fields and labelled by projective representations of the fundamental group,
are made in Theorem \ref{thm.gauge3}.

\section{Mathematical background.}
\label{sec.back}

In the present section we recall some mathematical tools that shall be used in the sequel.
In the following lines we start by defining some notations.

Given a Hilbert space $H$, we write $UH$ for the unitary group and $BH$ for the \sC algebra of bounded operators on $H$.
In the sequel, families of Hilbert spaces $H = \{ H_o \}_{o \in I}$ will be considered
and we shall write $BH = \{ BH_o \}_{o \in I}$ for the corresponding families of \sC algebras;
no confusion should arise with the case of a single Hilbert space, since it should be clear from the context when
a family is being considered.

Given a $C^*$ or von Neumann algebra $R$, we denote the unitary group of $R$ by $UR$.

Given a Hilbert space $H$ and a family $R = \{ R_o \}_{o \in I}$ of \sC subalgebras of $BH$, we denote the smallest \sC subalgebra of $BH$
containing $\cup_o R_o$ by 
$C^* \{ R_o \}_{o \in I}$.

\subsection{Geometry of posets.}
\label{sec.back.posets}

Let $M$ be a connected manifold. 
We say that a base $\Delta$ for the topology of $M$ is \emph{good} whenever 
any $a \in \Delta$ has compact closure and is arcwise and simply connected.
For example, we have the Minkowski spacetime, that has a good base given by the set of double cones,
or, more generally, a globally hyperbolic spacetime, that has a good base base given by the set of diamonds (\cite{GLRV01,BR08}).
In the present paper we shall always consider good bases.

Now, $\Delta$ becomes a \emph{poset} (partially ordered set) when endowed with the 
order relation given by the inclusion, and encodes some geometric invariants of $M$.
As a first step we consider the simplicial set defined by $\Delta$, 
whose sets at lower degrees are defined as follows (\cite{RR06,RRV07}):
\begin{itemize}
\item $\Sigma_0(\Delta) := \Delta$, the set of $0$-simplices, that are the analogues of "points";
\item $\Sigma_1(\Delta)$, the set of $1$-simplices, is the set of "oriented lines" of the type
\[
b := ( \partial_1b , \partial_0b ; |b| ) 
\ \ : \ \ 
\partial_1b , \partial_0b , |b| \in \Delta \ \ {\mathrm{and}} \ \ 
\partial_1b , \partial_0b \subseteq |b| \ .
\]
\item Finally we define the set $\Sigma_2(\Delta)$ of $2$-simplices ("oriented triangles") 
with elements quadruples
$c = ( \partial_0c , \partial_1c , \partial_2c \in \Sigma_1(\Delta) ; |c| \in \Delta)$
such that
\begin{equation}
\label{eq.dd}
\partial_{hk}c := \partial_h \partial_k c = \partial_k \partial_{h+1}c
\ \ , \ \
|\partial_hc| \subseteq |c|
\ \ , \ \
\forall h \geq k
\ .
\end{equation}
\end{itemize}
Using $1$-simplices we define \emph{paths}, in terms of sequences
\[
p := b_n * \ldots * b_1
\ \ : \ \ 
\partial_0b_k = \partial_1b_{k+1}
\ , \
\forall k = 1  , \ldots , n \ .
\]
To emphasize the initial and ending "points" of $p$ we use the notation
\[
p : a \to o \ \ , \ \ a = \partial_1b_1 \ , \ o = \partial_0b_n \ .
\]
The \emph{opposite path} of $p$ is given by 
$\ovl p := \ovl b_1 * \ldots * \ovl b_n : o \to a$, 
where any $\ovl b_k$, $k = 1 , \ldots , n$, is the 1-simplex defined by 
\[
\partial_1 \ovl b_k := \partial_0 b_k \ \ , \ \
\partial_0 \ovl b_k := \partial_1 b_k \ \ , \ \
| \ovl b_k | := | b_k | \ .
\]
A \emph{path frame} with pole $\omega \in \Delta$ is a set 
\begin{equation}
\label{def.pf}
P_\omega = \{ p_{a \omega} : \omega \to a \ , \ \forall a \in \Delta \}
\end{equation}
such that $p_{\omega \omega} = b_\omega := (\omega,\omega;\omega) \in \Sigma_1(\Delta)$. 
To be concise we write
\[
p_{\omega a} := \ovl p_{a \omega} \ .
\]
Paths of the type $p : a \to o$, $p' : o \to e$ can be composed by defining
$p'*p := b'_m * \ldots * b'_1 * b_n * \ldots b_1$,
so in particular the set of paths of the type $p : a \to a$, called \emph{loops}, is a semigroup.
There is an equivalence relation on paths, called \emph{homotopy}, and the quotient
$\pi_1(\Delta)$
of the set of loops starting and ending at $a \in \Delta$ is a group, 
whose isomorphism class does not depend on $a$.
There is an isomorphism
\begin{equation}
\label{eq.back1}
\pi_1(\Delta) \ \simeq \ \pi_1(M) \ ,
\end{equation}
where $\pi_1(M)$ is the fundamental group; this implies $\pi_1(M) = 0$
when $\Delta$ is directed (see \cite{Ruz05}).

\

Finally, in the spirit of \cite{Qui}, we mention that the above constructions apply to a category $C$ in place of a poset
{\footnote{This approach has been object of (unpublished) work by J. Roberts, G. Ruzzi and the author.}},
as for example the category of globally hyperbolic spacetimes at the root of locally covariant quantum field theory \cite{BFV03}.
One defines $\Sigma_0(C) = \obj C$, $\Sigma_1(C)$ as the set of pairs $b=(b_1,b_0)$ of arrows having the same target $|b| \in \obj C$
and so on, obtaining the notion of fundamental group of $C$.

\paragraph{Characteristic classes for representations of $\pi_1(\Delta)$.}
Let $G$ be a topological group and
\[
\chi : \pi_1(\Delta) \to G
\]
a morphism. Using (\ref{eq.back1}) we identify $\pi_1(\Delta)$ with the fundamental group and 
recall that the universal cover $\tilde M$ defines a principal $\pi_1(\Delta)$-bundle
$q : \tilde M \to M$,
so any fibre $q^{-1}(x)$, $x \in M$, is a right $\pi_1(\Delta)$-space.
Given a left $G$-space $F$ with action $g,v \mapsto gv$, $g \in G$, $v \in F$, this yields the induced $F$-bundle
\begin{equation}
\label{eq.back2}
q_\chi : P_\chi \to M \ ,
\end{equation}
where $P_\chi$ is the quotient of $\tilde M \times F$ by the equivalence relation
\[
(y,v) \sim (yp,\chi(p)v) \ \ , \ \ p \in \pi_1(\Delta) 
\ ,
\] 
and $q_\chi (y,v)_\sim := q(y)$, for any $(y,v)_\sim \in P_\chi$.
A $F$-bundle of the above kind is called \emph{locally constant bundle}, see \cite[\S I.2]{Kob}.

\

Let now $F$ be a manifold and $G$ a Lie group. 
We consider a $F$-bundle $P \to M$, not necessarily locally constant, and
a characteristic class $\zeta$ defining the $2k$-form
$\zeta(P) \in Z^{2k}_{deRham}(M)$, $k \in \bN$.
We say that $\zeta$ has periods in the ring $\bK \subset \bR$ whenever
\[
\int_c \zeta(P) \ \in \bK
\]
for any $2k$-cycle $c \in Z_{2k}(M)$.
In \cite{CS85} it is proved that there is a unique morphism
\[
\zeta^\uparrow(P) : C_{2k-1}(M) \to \bR / \bK \ ,
\]
where $C_{2k-1}(M)$ is the group of singular $(2k-1)$-chains, such that 
\[
\{ \zeta^\uparrow(P) \} (\partial \ell) \ = \ \int_\ell \zeta(P) \, {\mathrm{mod}} \bK 
\ \ , \ \
\forall \ell \in C_{2k}(M)
\ ,
\]
where $\partial : C_{2k}(M) \to Z_{2k-1}(M)$ is the boundary map.

Since it is well-known that $\zeta(P_\chi) = 0$ for any locally constant bundle
$P_\chi$
(see \cite[\S II.3]{Kob}), we have that $\zeta^\uparrow(P_\chi)$ vanishes on $\partial C_{2k}(M)$,
so it yields a cocycle defining a class in singular cohomology,
\begin{equation}
\label{eq.ccs0}
\zeta^\uparrow(\chi) := [\zeta^\uparrow(P_\chi)] \in H^{2k-1}(M,\bR / \bK) \ ,
\end{equation}
which by construction vanishes when $\chi$ is the trivial morphism.
For example, it is well-known that the Chern classes 
$c_k(P) \in Z^{2k}_{deRham}(M)$, $k \in \bN$, 
have periods in $\bZ$, so for any representation 
$\chi : \pi_1(\Delta) \to \ud$
there are classes
\begin{equation}
\label{eq.ccs}
c_k^\uparrow(\chi) \in H^{2k-1}(M,\bR / \bZ)
\ \ , \ \
k = 1 , \ldots , d \ .
\end{equation}

\subsection{Nets of \sC algebras.}
\label{sec.back.nets}

A \emph{net} 
{\footnote{Here we should use the term \emph{precosheaf}, nevertheless we prefer to maintain 
           the traditional terminology.}}
of \sC algebras over $\Delta$ is given by a family $A = \{ A_a \}_{a \in \Delta}$ of unital
\sC algebras, and unital *-monomorphisms (called \emph{the net structure})
\[
\jmath_{a'a} : A_a \to A_{a'}
\ \ : \ \
\jmath_{a''a'} \circ \jmath_{a'a} = \jmath_{a''a}
\ \ , \ \
\forall a \subseteq a' \subseteq a''
\ .
\]
We shall use the notation $\efA = (A,\jmath)_\Delta$. Morphisms 
$\nu : \efA \to \efA'$
of nets of \sC algebras are defined by families of *-morphisms
\[
\nu_a : A_a \to A'_a
\ \ : \ \
\nu_{a'} \circ \jmath_{a'a} = \jmath'_{a'a} \circ \nu_a
\ \ , \ \
\forall a \subseteq a'
\ .
\]
In algebraic quantum field theory it is customary to consider nets such that there is 
a Hilbert space $H$ with $A_a \subseteq BH$ for any $a \in \Delta$, and any
$\jmath_{a'a}$, $a \subseteq a'$,
is the inclusion map. In this case we write
\[
\efA = (A,\jmath)_\Delta \subset BH \ .
\]
\paragraph{Net bundles (\cite{RRV07,RRV08,RV11}).}
When every $\jmath_{a'a}$ is an isomorphism we say that $\efA$ is a \sC \emph{net bundle};
picking a family $\nu$ of *-isomorphisms
$\nu_a : A_a \to A_*$, $\forall a \in \Delta$,
and defining
\[
\jmath'_{a'a} \in {\bf aut}A_*
\ \ , \ \
\jmath'_{a'a} := \nu_{a'} \circ \jmath_{a'a} \circ \nu_a^{-1}
\ \ , \ \
\forall a \subseteq a'
\ ,
\]
we obtain a \sC net bundle $\efA' = (A',\jmath')_\Delta$ with constant fibre $A'_a \equiv A_*$ and
isomorphic to $\efA$. We call $\nu$ a \emph{choice of the standard fibre} of $\efA$,
and denote the set of \sC net bundles with fibre isomorphic to $A_*$ by
\[
{\bf bun}(\Delta,A_*) \ .
\]
To be concise we write $\jmath_{aa'} := \jmath_{a'a}^{-1}$, $\forall a \subseteq a'$.
For any $b \in \Sigma_1(\Delta)$ we define
$\jmath_b := \jmath_{\partial_0b \, |b|} \circ \jmath_{|b| \, \partial_1b}$, 
and, assuming that a choice of the standard fibre $A_*$ has been made,
\begin{equation}
\label{eq.back3a}
\chi_p :=  \jmath_{b_n} \circ \ldots \circ \jmath_{b_1} \ \in {\bf aut}A_*
\ \ , \ \
\forall p = b_n * \ldots * b_1 : a \to a
\ .
\end{equation}
The map $\chi_p$, $p : a \to a$, factorizes through homotopy, so it yields the \sC dynamical system 
$\chi : \pi_1(\Delta) \to {\bf aut}A_*$.
Indeed, assigning the \emph{holonomy}
\begin{equation}
\label{eq.back3}
\efA \mapsto \chi
\end{equation}
defines an equivalence from the category of \sC net bundles to the one of $\pi_1(\Delta)$-dynamical \sC systems.

Analogous definitions, notations and properties hold in the cases of group net bundles and Hilbert net bundles.
In the first case we have group isomorphisms 
$\imath_{a'a}: G_a \to G_{a'}$, $\forall a \subseteq a'$,
defining the \emph{group net bundle} $\efG = (G,\imath)_\Delta$,
whilst in the second case we have unitaries
$U_{a'a}: H_a \to H_{a'}$, $\forall a \subseteq a'$,
defining the \emph{Hilbert net bundle} $\efH = (H,U)_\Delta$.
We have the composition of equivalences
\begin{equation}
\label{eq.hol}
\efP \ \stackrel{ (\ref{eq.back3}) }{\longrightarrow} \ 
\chi \ \stackrel{ (\ref{eq.back2}) }{\longrightarrow} \ 
P_\chi \ ,
\end{equation}
assigning to any net bundle $\efP$ the locally constant bundle $P_\chi$,
which holds when the fibres are \sC algebras, Hilbert spaces and topological groups.
The geometric content of (\ref{eq.hol}) is that $\chi$ is the holonomy of $\efP$
or, equivalently, the monodromy of $P_\chi$,
whilst $\efP$ can be understood as a precosheaf of local sections of $P_\chi$.

\paragraph{Representations (\cite{BR08,RV11}).} 
Let $\efA = (A,\jmath)_\Delta$ be a net of \sC algebras.
A \emph{representation} of $\efA$ is a pair $(\pi,U)$, where 
$H = \{ H_a \}_{a \in \Delta}$
is a family of Hilbert spaces,
$U = \{ U_{a'a} : H_a \to H_{a'} \}_{a \subseteq a'}$
is a family of unitary operators and
$\pi_a : A_a \to BH_a$, $a \in \Delta$,
is a family of Hilbert space representations such that
\[
U_{a''a} = U_{a''a'} U_{a'a} 
\ \ , \ \
\ad U_{a'a} \circ \pi_a = \pi_{a'} \circ \jmath_{a'a}
\ \ , \ \
a \subseteq a' \subseteq a'' 
\ .
\]
Defining the Hilbert net bundle $\efH = (H,U)_\Delta$ we obtain the \sC net bundle
$\efB \efH = (BH,\ad U)_\Delta$ and $(\pi,U)$ can be regarded as a morphism
$\pi : \efA \to \efB \efH$.
When $\efH$ is trivial 
{\footnote{That is, $\efH$ is isomorphic to the constant Hilbert net bundle 
$\efH^0 = (H^0,U^0)_\Delta$,
$H^0_a \equiv H_*$, $U^0_{a'a} \equiv 1_{H_*}$.}}
we have Hilbert space representations, that are those commonly used in algebraic quantum field theory.

\paragraph{Commutators on Hilbert net bundles.}
We give a notion of commutator for operators on Hilbert net bundles, 
designed to be coherent with the underlying net structure and not depending, like the one used in \cite[Eq.3.6]{BR08},
on the choice of the standard fibre.
We shall use this notion to propose a refined notion of causality for nets in generic spacetimes (see \S \ref{sec.gauge}).

Let $\efH = (H,U)_\Delta$ be a Hilbert net bundle.
We say that $\efH$ is trivial on $\Omega \subseteq \Delta$ whenever the restriction 
$\efH_\Omega := (H,U)_\Omega$
is a trivial Hilbert net bundle.
Given $a,o \in \Delta$, we say that a path 
$p : o \to a$, $p=b_n * \ldots * b_1$,
is \emph{trivialized} by $\efH$ whenever there is $\Omega \subseteq \Delta$ such that:
(1) $\efH$ is trivial on $\Omega$;
(2) $b_i \in \Sigma_1(\Omega)$, $\forall i=1,\ldots,n$.
We denote the set of paths trivialized by $\efH$ by $\pi_0(\Delta,\efH)$.

Given $a,o \in \Delta$ and $T \in BH_a$, $T' \in BH_o$, we write 
$[T,T']^{net}_\pm = 0$ 
whenever
\begin{equation}
\label{eq.nets1}
[T,T']^p_\pm \ := \ T \cdot U_p T'U_p^* \pm U_p T'U_p^* \cdot T \ = \ 0 
\ \ , \ \
\forall p : o \to a \ , \ p \in \pi_0(\Delta,\efH)
\ ,
\end{equation}
where $U_p := U_{b_n} \cdots U_{b_1}$, $p = b_n * \ldots * b_1$
(recall that we set $U_b := U_{\bo |b|} U_{|b| \bl}$, $\forall b \in \Sigma_1(\Delta)$).
Note that the previous expression makes sense since $U_p T'U_p^* \in BH_a$. 
Moreover (\ref{eq.nets1}) is symmetric in the sense that $[T,T']^p_\pm = 0$ if, and only if,
\[
{}^p[T,T']_\pm \ := \ U_{\ovl p}TU_{\ovl p}^* \cdot T' \pm T' \cdot U_{\ovl p}TU_{\ovl p}^* \ = \ 0 \ ,
\]
where $\ovl p \in \pi_0(\Delta,\efH)$, $\ovl p : a \to o$, is the opposite path.
Note that: 
\begin{itemize}
\item the fact that $p$ is trivialized by $\efH$ means that there is a Hilbert space $H'$ 
      with unitaries $V_e : H_e \to H'$, $e \in \Omega$, such that 
      $U_{e'e} = V_{e'}^*V_e$, $\forall e \subseteq e' \in \Omega$.
      So we have
      \[
      U_p \ = \ V_a^* V_{|b_n|} \cdots V_{|b_1|}^* V_o \ = \ V_a^* V_o \ ,
      \]
      and $U_p$ is independent of the choice of the path $p$ with 1-simplices in $\Sigma_1(\Omega)$.
      When $a=o$ we have $U_p=1$ for any $p \in \pi_0(\Delta,\efH)$, $p : a \to a$;
      thus $[T,T']^{net}_\pm$ is the usual (anti)commutator of $BH_a$.
\item if $\efH$ is trivial then any path is trivialized by $\efH$. In particular, when
      $\efH$ is the constant Hilbert net bundle, $U_p \equiv 1$ and we have the usual (anti)commutator;
\item if $\Omega \subseteq \Delta$ is such that $\pi_1(\Omega)= {\bf 0}$ and $p$ is made of 1-simplices in $\Sigma_1(\Omega)$,
      then $p$ is automatically trivialized by $\efH$; for example, we may take $\Omega$ directed.
      Note that restrictions to simply connected
      subsets of $\Delta$ appear in \cite[\S 3.2]{BR08} to define sharp excitations of the vacuum state.
\end{itemize}

\paragraph{Gauge actions (\cite{Vas12}).}
Let $\efA = (A,\jmath)_\Delta$ be a net of \sC algebras and $\efG = (G,\imath)_\Delta$
a group net bundle (we assume that the fibres of $\efG$ are locally compact and Hausdorff).
A \emph{gauge action} of $\efG$ on $\efA$ is given by a family of strongly continuous actions
$\alpha_o : G_o \to {\bf aut}A_o$, $\forall o \in \Delta$,
such that
\begin{equation}
\label{eq.GA00}
\{ \alpha_{o'} (\imath_{o'o} (g)) \} \circ \jmath_{o'o}
=
\jmath_{o'o} \circ \alpha_o (g)
\ \ , \ \
\forall g \in G_o
\ , \
o \subseteq o'
\ .
\end{equation}
When $\efG$ is the constant net bundle (that is, any $\imath_{o'o}$ is the identity of a
fixed group $G_*$), we have a group action of $G_*$ on $\efA$ by net automorphisms.
A morphism $\nu : \efA \to \efA'$ is said to be $\efG$-\emph{equivariant} whenever
\[
\nu_a \circ \alpha_a(g) \ = \ \alpha'_a(g) \circ \nu_a
\ \ , \ \
\forall a \in \Delta \ , \ g \in G_a \ .
\]
The notion of gauge action can be defined for a Hilbert net bundle
$\efH = (H,U)_\Delta$
as above, requiring that each 
$\alpha_o(g)$, $o \in \Delta$, $g \in G_o$, 
is a unitary operator on $H_o$.
The category of $\efG$-equivariant Hilbert net bundles over $\Delta$, that we denote by
\[
{\bf Hilb}_\efG(\Delta) \ ,
\]
is a symmetric tensor \sC category (see \S \ref{sec.back.presheaves} for an explanation of these terms);
it has a simple unit in the cases of interest, that is, when $\Delta$ is path-connected or, equivalently, 
when $M$ is arcwise connected (see \cite[\S 5]{RRV08}).

\subsection{Nets and presheaves of tensor categories.}
\label{sec.back.presheaves}

In the present section we recall some properties of nets and presheaves of tensor categories,
that constitute one of the main objects of study in the present paper. 
More detailed references on this topic are \cite{Rob,Vas12}.

\paragraph{DR-categories.}
A \sC \emph{category} is a category $C$ whose spaces of arrows are Banach spaces
endowed with an antilinear involution
$* : (\rho,\sigma) \to (\sigma,\rho)$, $\forall \rho,\sigma \in \obj C$,
such that
$(t \circ t')^* = t'^* \circ t^*$, $\| t^* \circ t \| = \| t \|^2$
for any pair of composable arrows $t,t'$. 
Since the composition of arrows is assumed to be linear, 
this implies that any $(\rho,\rho)$, $\rho \in \obj C$, is a \sC algebra.

We assume that functors between \sC categories preserve the involution structure.
Following the terminology of \cite[\S I.3]{ML}, we say that a functor: 
is an \emph{embedding}, whenever it is faithful on the spaces of arrows;
is \emph{full}, whenever it is surjective on the spaces of arrows;
is an \emph{isomorphism}, whenever it has an inverse functor and, in particular, 
an \emph{automorphism}, whenever the initial and target categories coincide.

A \emph{tensor product} on $C$ is given by a bilinear bifunctor
$\otimes : C \times C \to C$,
and the \emph{identity object} $\iota \in \obj C$ is characterized by the property
that $\iota \otimes \rho = \rho \otimes \iota = \rho$ for any $\rho \in \obj C$.
To be concise, in the sequel we shall write $\rho \sigma := \rho \otimes \sigma$,
$\forall \rho,\sigma \in \obj C$.
We say that $C$ has a \emph{simple unit} whenever $(\iota,\iota) = \bC 1_\iota$, 
where $1_\iota$ is the identity arrow. 
We say that the tensor product is \emph{symmetric} whenever there is a family
$\eps_{\rho,\sigma} \in (\rho \sigma , \sigma \rho)$
of unitary arrows, inducing the flip 
$(t \otimes t') \circ \eps_{\rho,\rho'} = \eps_{\sigma,\sigma'} \circ(t' \otimes t)$, 
$t \in (\rho,\sigma)$, $t' \in (\rho',\sigma')$,
and fulfilling the properties in \cite[\S 1]{DR89}.
In this case we write $C_{\otimes,\eps}$.

The category of finite dimensional, unitary representations (that is, the \emph{dual}) of a compact group
is the basic class of examples for a symmetric tensor \sC category with simple unit.
A symmetric tensor \sC category with simple unit, conjugates, direct sums and subobjects
in the sense of \cite[\S 1]{DR89} is said to be a \emph{DR-category}, 
and is automatically isomorphic to the dual of a compact group, unique up to isomorphism (\cite[Theorem 6.9]{DR89}).

\paragraph{Nets of tensor categories.}
A \emph{net of tensor \sC categories} over $\Delta$ is given by a family $T = \{ T_a \}_{a \in \Delta}$
of tensor \sC categories and a family of tensor embeddings
\[
\jmath_{a'a} : T_a \to T_{a'}
\ \ : \ \
\jmath_{a''a'} \circ \jmath_{a'a} = \jmath_{a''a}
\ \ , \ \
\forall a \subseteq a' \subseteq a''
\ .
\]
We shall use the notation $\efT_\otimes = (T,\jmath)_\Delta$. 
We say that $\efT_\otimes$ is \emph{simple} whenever every $T_a$, $a \in \Delta$, has a simple unit.
When every $\jmath_{a'a}$ is an isomorphism we say that $\efT_\otimes$ is a \emph{net bundle} of tensor \sC categories.
When every $T_a$, $a \in \Delta$, has symmetry $\vareps_a$ and 
$\jmath_{a'a}(\vareps_{a;\rho,\sigma}) = \vareps_{a'; \jmath_{a'a}(\rho) , \jmath_{a'a}(\sigma)}$
for any $a \subseteq a'$, $\rho,\sigma \in \obj T_a$, we say that $\efT_\otimes$ is
\emph{symmetric} and write $\efT_{\otimes,\vareps}$.

The next class of examples is motivated by the analysis of superselection structures
in algebraic quantum field theory.
\begin{ex}{(See \cite[\S 27]{Rob} and \cite[\S 3.3]{GLRV01}).}
\label{ex.Rob}
{\it
Let $A$ be a unital \sC algebra. Then the semigroup ${\bf end}A$ of unital *-endomorphisms of $A$
defines a tensor \sC category with sets of arrows
\begin{equation}
\label{eq.intro0}
(\rho,\sigma) := \{ t \in A \ : \ t \rho(v) = \sigma(v)t \ , \ \forall v \in A \}
\ \ , \ \ 
\forall \rho,\sigma \in {\bf end}A
\ ,
\end{equation}
whose composition is given by the product of $A$; the tensor product is given by
\begin{equation}
\label{eq.intro1}
\rho \otimes \sigma := \rho \sigma := \rho \circ \sigma
\ \ , \ \
t \otimes t' := t \rho(t')
\ \ , \ \
\forall \rho,\sigma,\rho',\sigma' \in {\bf end}A
\ , \ 
t \in (\rho,\sigma) \ , \  t' \in (\rho',\sigma')
\ .
\end{equation}
Let now $H$ be a Hilbert space and $\efR = (R,\jmath)_\Delta \subset BH$ a net of \sC algebras on $H$. 
For any $a \in \Delta$ we consider the \sC algebra
\[
R^a \, := \, C^* \{ R_o \}_{o \supseteq a } \subseteq BH
\]
(note that $R_a \subseteq R^a$)
and, given $\rho,\sigma \in {\bf end}R^a$, the space $(\rho,\sigma) \subseteq R^a$ of (\ref{eq.intro0}). 
We define the categories
\[
T_a := 
\left\{
\begin{array}{ll}
\obj T_a \ := \ \{ \rho \in {\bf end}R^a : \rho(R_o) \subseteq R_o , \forall o \supset a \} \ ,
\\
{\bf arr}T_a \ := \ \{ (\rho,\sigma)_a := R_a \cap (\rho,\sigma) \ , \ \forall \rho,\sigma \in \obj T_a  \} \ ,
\end{array}
\right.
\]
having tensor product (\ref{eq.intro1}).
We have 
\[
R_a \subseteq R_{a'} \ \ , \ \ R^{a'} \subseteq R^a 
\ \ , \ \
\forall a \subseteq a' \ ,
\]
and by definition any $\rho \in \obj T_a$ restricts to an endomorphism 
$\jmath_{a'a}\rho \in \obj T_{a'}$. 
Moreover we have inclusions
\[
\jmath_{a'a} : (\rho,\sigma)_a \to ( \jmath_{a'a}\rho , \jmath_{a'a}\sigma )_{a'}
\ \ , \ \
\rho,\sigma \in \obj T_a \ ,
\]
and this yields the net $\efT = (T,\jmath)_\Delta$.
Finally, it is trivial to verify that $\jmath$ preserves the tensor structure, so $\efT_\otimes$ is
a tensor net. Note that the identity of $T_a$ is the identity automorphism $\iota_a \in {\bf end}R^a$, so
$(\iota_a,\iota_a)_a = R_a \cap (R^a)'$.
}
\end{ex}

\paragraph{Presheaves.} There is a natural notion dual to the one of net:
a \emph{presheaf} of \sC categories is given by a family of embeddings
\[
r_{aa'} : T_{a'} \to T_a
\ \ : \ \
r_{aa'} \circ r_{a'a''} = r_{aa''} 
\ \ , \ \
\forall a \subseteq a' \subseteq a''
\ .
\]
We shall write $\efT = (T,r)^\Delta$. 
A \emph{presheaf morphism} $\eta : \efT \to \efT'$ is a family of functors
\[
\eta_a : T_a \to T'_a
\ \ : \ \
\eta_a \circ r_{aa'} = r'_{aa'} \circ \eta_{a'}
\ \ , \ \
\forall a \subseteq a' \ .
\]
Notation and terminology for presheaves are analogous to the ones of nets, 
so we say that $\efT_\otimes$ is \emph{simple} whenever every $T_a$, $a \in \Delta$, is simple, and
write $\efT_{\otimes,\vareps}$ to indicate presheaves endowed with a symmetric tensor structure. 
In particular, we say that $\eta$ is a \emph{symmetric tensor} presheaf morphism whenever
the involved structures are preserved by $\eta$.

We say that $\efT$ is \emph{full} whenever every $r_{aa'}$, $a \subseteq a'$, is a full functor.
In particular, a \emph{presheaf bundle} is a presheaf $\efT = (T,r)^\Delta$ where any
$r_{aa'}$ is an isomorphism. The map
\[
\efT = (T,\jmath)_\Delta \ \mapsto \ \efT' = (T,r)^\Delta \ , \ 
r_{aa'} := \jmath_{a'a}^{-1} \ , \ 
\forall a \subseteq a' \ ,
\]
yields an isomorphism from the category of net bundles to the one of presheaf bundles.
Given a symmetric tensor \sC category $T_{\otimes,\eps}$, we denote
the set of symmetric tensor presheaf bundles with fibre $T_{\otimes,\eps}$ by
${\bf pbun}^\uparrow(\Delta,T_{\otimes,\eps})$.
Reasoning as in (\ref{eq.back3a}), we obtain the one-to-one correspondence
\[
{\bf pbun}^\uparrow(\Delta,T_{\otimes,\eps}) \ \lra \ \hom( \pi_1(\Delta) , {\bf aut}T_{\otimes,\eps} )
\ \ , \ \
\efT_{\otimes,\vareps} \mapsto \chi \ ,
\]
where ${\bf aut}T_{\otimes,\eps}$ is the group of symmetric tensor automorphisms of $T_{\otimes,\eps}$.
We call $\chi$ the \emph{holonomy of} $\efT_{\otimes,\vareps}$.

\begin{ex}
\label{ex.back1}
{\it
Let $d \in \bN$ and $\pi_G : G \to \ud$ denote the defining representation of a compact Lie group $G \subseteq \ud$.
We denote the category of tensor powers of $\pi_G$ by $\wa \pi_G$, which is a tensor \sC category with symmetry $\theta$. 
The group $G$ defines a \sC algebra $\mO_G$, the fixed-point algebra of the natural $G$-action on the Cuntz algebra $\mO_d$ (\cite{DR87}).
We define
\[
QG \ := \ NG / G \ ,
\]
where $NG \subseteq \ud$ is the normalizer of $G$. 
Since there are maps
\[
QG \to {\bf aut}\wa \pi_{G;\otimes,\theta}
\ \ , \ \
QG \to {\bf aut}\mO_G
\]
(see \cite[\S 3]{Vas09}), by the previous considerations we obtain maps
\begin{equation}
\label{eq.FL08}
{\hom}( \pi_1(\Delta) , QG ) \ \to \ {\bf pbun}^\uparrow(\Delta,\wa \pi_{G;\otimes,\theta}) 
\ \ , \ \
{\hom}( \pi_1(\Delta) , QG )	\ \to \ {\bf bun}(\Delta,\mO_G) 
\ .
\end{equation}
}
\end{ex}

\paragraph{Sections.} By definition, the \emph{category of sections} of the symmetric tensor presheaf 
$\efT_{\otimes,\vareps} = (T,r)^\Delta$
has objects and, respectively, arrows
\[
\left\{
\begin{array}{lllll}
\varrho = \{ \varrho_a \in \obj T_a \}_{a \in \Delta} 
& : &
r_{aa'}(\varrho_{a'}) = \varrho_a 
& , & 
\forall a \subseteq a' \ ,
\\
t = \{ t_a \in (\varrho_a,\varsigma_a) \}_{a \in \Delta}
& : &
r_{aa'}(t_{a'}) = t_a
& , &
\forall a \subseteq a' \ .
\end{array}
\right.
\]
We denote it by $\wt \efT_{\otimes,\vareps}$. It is easily verified that $\wt \efT_{\otimes,\vareps}$
is a symmetric tensor \sC category, and we say that $\efT_{\otimes,\vareps}$ is a \emph{DR-presheaf}
whenever $\wt \efT_{\otimes,\vareps}$ is a DR-category.

\begin{rem}
\label{rem.back1}
Given the symmetric tensor presheaf $\efT_{\otimes,\vareps} = (T,r)^\Delta$, 
for any $a \in \Delta$ we denote the full subcategory of $T_a$ with objects
$\varrho_a$, $\varrho \in \obj \wt \efT_{\otimes,\vareps}$,
by $T^\downarrow_a$. Any $r_{aa'}$, $a \subseteq a'$, restricts to an embedding 
$r^\downarrow : T^\downarrow_{a'} \to T^\downarrow_a$, $\forall a \subseteq a'$,
and this yields the symmetric tensor presheaf 
$\efT^\downarrow_{\otimes,\vareps} = (T^\downarrow,r^\downarrow)^\Delta$. 
By construction the section categories of 
$\efT_{\otimes,\vareps}$, $\efT^\downarrow_{\otimes,\vareps}$ 
coincide, that is
$\wt \efT_{\otimes,\vareps}  =  \wt \efT^\downarrow_{\otimes,\vareps}$.
\end{rem}

To illustrate a Tannaka duality for categories of sections we give the notion
of \emph{embedding}, that is, a symmetric tensor presheaf morphism
\[
I \ : \ \efT_{\otimes,\vareps} = (T,r)^\Delta \  \to \ \efC_{\otimes,\vartheta} = (C,r')^\Delta \ ,
\]
such that: 
(1) $\efC_{\otimes,\vartheta}$ is a full presheaf;
(2) any $C_a$, $a \in \Delta$, is a full symmetric tensor subcategory of the one of finite-dimensional Hilbert spaces;
(3) any $I_a$, $a \in \Delta$, is an embedding.
\begin{thm}[Tannaka duality, \cite{Vas12}]
\label{thm_DUAL}
Let $\efT_{\otimes,\vareps} = (T,r)^\Delta$ be a simple, symmetric tensor presheaf having an embedding. 
Then there is a group net bundle $\efG = (G,i)_\Delta$ with a full symmetric tensor embedding 
$\wt \efT_{\otimes,\vareps} \to {\bf Hilb}_\efG(\Delta)$.
\end{thm}

Now, any section $\varrho \in \obj \wt \efT_{\otimes,\vareps}$ defines a presheaf bundle
${}_\beta \wa \varrho = ( \wa \varrho , r^\varrho )_\Delta$,
constructed as follows.
Given $a \in \Delta$, the fibre $\wa \varrho_a$ has objects the tensor powers 
$\varrho_a^n = \varrho \otimes \ldots \otimes \varrho$, $n \in \bN$,
and arrows
\begin{equation}
\label{eq.beta}
{}_\beta (\varrho_a^n , \varrho_a^m ) 
\ := \
\{ t_a \ : \ t \in (\varrho^n , \varrho^m )_{\geq a}  \} \ ,
\end{equation}
where, for any $\varrho , \varsigma \in \obj \wt \efT_{\otimes,\vareps}$, we use the notation
\[
(\varrho , \varsigma )_{\geq a} \ := \ 
\{  t = \{ t_o \in (\varrho_o,\varsigma_o) \}_{o \supseteq a} \ : \ 
    r_{oo'}(t_{o'}) = t_o \ , \ \forall o \subseteq o'  \}
\]
(see \cite[Prop.3.4]{Vas12} for details). The restriction functors 
$r^\varrho_{aa'} : \wa \varrho_{a'} \to \wa \varrho_a$, $a \subseteq a'$,
are defined by the identity on the objects and by restricting $r_{aa'}$ to $\wa \varrho_{a'}$ on the arrows.
The following result describes ${}_\beta \wa \varrho$ in terms of the presheaf bundles of Ex.\ref{ex.back1}:
\begin{thm}[Holonomy, \cite{Vas12}]
\label{lem_CCS00}
Let $\efT_{\otimes,\vareps} = (T,r)^\Delta$ be a DR-presheaf. 
Then for any section $\varrho \in \obj \wt \efT_{\otimes,\vareps}$ there are a compact Lie group
$G_\varrho \subseteq \ud$
such that, up to isomorphism of the fibre,
\[
{}_\beta \wa \varrho \ \in \ {\bf pbun}^\uparrow(\Delta,\wa \pi_{G_\varrho;\otimes,\theta})
\]
and a morphism
\[
\chi_\varrho : \pi_1(\Delta) \to QG_\varrho := NG_\varrho / G_\varrho \ ,
\]
which is the holonomy of ${}_\beta \wa \varrho$. If there is an embedding of $\efT_{\otimes,\vareps}$ then
for any section $\varrho$ there must be a lift 
$\wa \chi_\varrho : \pi_1(\Delta) \to NG_\varrho$, 
$\chi_\varrho = \wa \chi_\varrho \mod G_\varrho$.
\end{thm}

The existence of embeddings can be described in cohomological terms:
\begin{thm}[Gerbes, \cite{Vas12}]
\label{thm.gerbe}
Any section $\varrho$ of a DR-presheaf defines the cochain
\[
u_\varrho : \Sigma_1(\Delta) \to NG_\varrho
\ \ : \ \
du_\varrho(c) := 
u_\varrho(\partial_0c) u_\varrho(\partial_2c) u_\varrho(\partial_1c)^* \in G_\varrho
\ , \ 
\forall c \in \Sigma_2(\Delta)
\ .
\]
(Eventual) lifts of $\chi_\rho$ are in one-to-one correspondence with pairs $(v,g)$,
$v : \Sigma_0(\Delta) \to NG_\varrho$,
$g : \Sigma_1(\Delta) \to  G_\varrho$,
such that 
\begin{equation}
\label{eq.gerbe}
z(b) := v(\partial_0b) g(b) u_\varrho(b) v(\partial_1b)^* \ \in NG_\varrho
\ \ , \ \
\forall b \in \Sigma_1(\Delta)
\ ,
\end{equation}
is a cocycle, i.e. $dz(c) = 1$, $\forall c \in \Sigma_2(\Delta)$.
\end{thm}

\section{Presheaves in quantum field theory.}
\label{sec.qft}

In the present section we give the main result of the present paper, 
showing that the superselection structures over globally hyperbolic spacetimes introduced by Brunetti and Ruzzi
(\cite{BR08}) can be interpreted as section categories of simple, symmetric tensor presheaves.

Actually, the fact that a superselection structure defines a {\em net} $\efT$ 
of tensor \sC categories was recognized in \cite[\S 3]{GLRV01}, \cite[\S 27]{Rob} (see Ex.\ref{ex.Rob}).
Nevertheless $\efT$ does not suffice to characterize the superselection structure itself
(see \cite[\S 3.3]{GLRV01}, \cite[\S 5.1]{BR08}),
so, rather than *-endomorphisms as in \cite{GLRV01,Rob}, 
we make use of cocycles, a point of view introduced by J.E. Roberts (\cite{Rob0}).

\paragraph{Sectors with nontrivial holonomy.}
Let $M$ be a globally hyperbolic spacetime. Then we have a {\em causal disjointness relation},
that is, a binary, symmetric relation $\perp$ defined on subsets of $M$, see \cite[\S 2.1]{GLRV01}.
The physical meaning of $\perp$ is that it encodes Einstein causality, which states that events localized in
causally disjoint regions of $M$ do not interfere one each other. 

In the following $\Delta$ shall denote the base of diamonds generating the topology of $M$
(see \cite[\S 2.1]{GLRV01} or, more shortly, \cite[\S 3.1]{BR08}), thus the abstract results of the previous section apply.
We note that any $a \in \Delta$ has proper, compact closure $\ovl{a}$. 
%
%
To be concise we write $e < a$ whenever the closure of $e \in \Delta$ is a proper subset 
of $a \in \Delta$, and set $\Delta^a := \{ e \in \Delta : e < a \}$. By \cite[Lemma B.5]{BR08},
for any $e \in \Delta^a$ there is $o \in \Delta^a$ with $o \perp e$.

Given the Hilbert space $H$, a {\em Haag-Kastler net} over $M$ is given by a net of von Neumann algebras 
$\efR = (R,j)_\Delta \subset BH$, 
fulfilling the following properties:
\begin{enumerate}
\item Einstein causality: $R_o \subseteq R_e'$ for any $e \perp o$;
\item Irreducibility: $\cap_a R_a' = \bC 1$;
\item Outer regularity: $R_a = \cap_{a < o} R_o$ for any $a \in \Delta$;
\item Borchers property: for any $a \in \Delta$, $e < a$ and projection $E \in R_e$
      there is an isometry $V_a \in R_a$ such that $E = V_aV_a^*$.
\item Punctured Haag duality: if $x \in M$, then we have
      $R_a = \cap \{ R_o' : o \in \Delta , o \perp a , \ovl o \perp \{ x \} \}$
      for any $a \in \Delta$ with $\ovl a \perp \{ x \}$.
\end{enumerate}
The crucial property from the viewpoint of physical significance is Einstein causality,
which states the principle that spacelike separated quantum observables must commute
and hence there is not uncertainty between them. 
The standard way to construct Haag-Kastler nets on curved spacetimes 
is the use of free fields (\cite{Dim80,Dim82,BFM09}).
About (punctured) Haag duality, see \cite{Ver97,Ruz05a} and related references.

\smallskip

We now define the $C^*$-category
$Z^1(\efR)$
with objects unitary cocycles
\[
z := \{ z(b) \in UR_{|b|} \}_{b \in \Sigma_1(\Delta)}
\ \ : \ \
dz(c) := z(\partial_0c) z(\partial_2c) z(\partial_1c)^* = 1
\ , \
\forall c \in \Sigma_2(\Delta)
\ ,
\]
and arrows
\[
(z,z')
\ := \
\{ \ t = \{ t_a \in R_a \}_{a \in {\Delta}} 
     \ : \
     t_{\partial_0b} z(b)  =  z'(b) t_{\partial_1b}
     \ , \
     \forall b \in \Sigma_1(\Delta) \ \} \ .
\]
For any path $p : a \to o$, we set
\[
z(p) := z(b_n) \cdots z(b_1)
\ \ , \ \
p = b_n * \ldots * b_1 \ .
\]
In particular, we have the identity cocycle $\iota(b) := 1$, $\forall b \in \Sigma_1(\Delta)$.
Note that if $t \in (\iota,\iota)$ then $t_a \equiv t_0$, $\forall a \in \Delta$,
with $t_0 \in \cap_a R_a$; we have 
\begin{equation}
\label{eq.Ap1}
\cap_{a \ni x} R_a = \bC 1 
\ \ , \ \
\forall x \in M
\end{equation}
(see \cite[\S 3.2]{BR08},\cite[Eq.23]{Ruz05}) and this implies $(\iota,\iota) = \bC 1$. 

\

As argued in \cite[\S 4]{BR08}, $Z^1(\efR)$ encodes a superselection structure 
that gives rise to quantum charges affected by the topology of $M$: 
for any $z \in \obj Z^1(\efR)$ we define $H^z_a := H$, $a \in \Delta$, and unitaries 
\[
U^z_{a'a} := z(a,a';a') : H^z_a \to H^z_{a'}
\ \ \Rightarrow \ \ 
U^z_{a''a'} U^z_{a'a} = U^z_{a''a}
\ , \
\forall a \subseteq a' \subseteq a''
\ .
\]
This defines the Hilbert net bundle $\efH^z := (H^z,U^z)_\Delta$,
which, on turns, yields the \sC net bundle $\efB \efH^z = ( BH^z , \ad U^z )_\Delta$
and the representation 
\begin{equation}
\label{eq.C00}
\pi^z : \efR \to \efB \efH^z
\ \ , \ \
\pi^z_a(t) := z(p) t z(p)^*
\ \ , \ \qquad
\forall a \in \Delta
\ , \
t \in R_a \ ,
\ {\mathrm{with}} \
p : a \to o
\ , \
o \perp a
\end{equation}
(see \cite[Appendix A]{BR08}). When $\efH^z$ is trivial
we have sectors in the classical sense (topologically trivial sectors),
that define representations of Fredenhagen's universal algebra $\vec{R}$ (\cite{Fre}).
When $M$ is the Minkowski spacetime $Z^1(\efR)$ is the well-known DHR-superselection structure (see \cite{Rob0}).

\smallskip

The conjugation structure on $Z^1(\efR)$,
which encodes the particle-antiparticle correspondence in the setting of superselection sectors,
is defined in the following standard way (\cite{DR89,DR90,GLRV01,Ruz05,BR08}):
a cocycle $\bar z \in \obj Z^1(\efR)$ is said to be a \emph{conjugate} of $z \in \obj Z^1(\efR)$ whenever there are
$r \in (\iota,\bar z \otimes z)$ and $\bar{r} \in (\iota,z \otimes \bar z)$
fulfilling the conjugate equation
\[
\bar{r}^* \otimes 1_z \cdot 1_z \otimes r \, = \, 1_z
\ \ , \ \
r^* \otimes 1_{\bar z} \cdot 1_{\bar z} \otimes \bar{r} \, = \, 1_{\bar z}
\]
(about the meaning of $\otimes$ we refer the reader to \cite[\S 5]{BR08}, anyway it should be clear by reading the next subsection).
In accord with the discussion of \cite[\S 5.3]{BR08},
the full subcategory $Z^{1,\bullet}(\efR)$ of $Z^1(\efR)$ with objects cocycles admitting
a conjugate is a DR-category and yields the set of sectors with finite statistics.

\smallskip

The rest of the section is devoted to the proof that $Z^1(\efR)$ is
the section category of a simple, symmetric tensor presheaf. 
This automatically gives to $Z^1(\efR)$ 
the structure of a symmetric tensor \sC category with simple unit, 
that, as we shall see, coincides with the one defined in \cite[\S 5]{BR08}.
For reader's convenience, and due to the fact that we shall work with restricted nets that
may not fulfill Haag duality, often we shall write explicitly standard computations.

\subsection{Charge structure.}
To introduce the charge structure we follow the standard way and, given $a \in \Delta$, define the set
${\bf loc}R_a \subseteq {\bf end}R_a$
whose elements are *-endomorphims
\begin{equation}
\label{eq.C0}
\rho \in {\bf end}R_a  \ \ | \ \  \exists e < a 
\ \ {\mathrm{fulfilling}} \ \
\left\{
\begin{array}{ll}
\rho(R_o) \subseteq R_o \ , \ \forall o \supseteq e ,  o < a ,
\\
\rho = id_{R_o} \ , \ \forall o \perp e ,  o < a \ .
\end{array}
\right.
\end{equation}
In the sequel we shall keep in evidence the \emph{localization region} $e \in \Delta^a$ by writing $\rho_e \equiv \rho$.
Note that
$\tau_e := \rho_e \circ \sigma_e \in {\bf loc}R_a$,
$\forall \rho_e,\sigma_e \in {\bf loc}R_a$.

We now consider the category $Z^1_c(\efR)_{< a}$ with objects pairs $(z,\rho) \equiv z_\rho$, where
\[
z = \{ z(b) \in UR_{|b|} \}_{b \in \Sigma_1(\Delta^a)} 
\ \ , \ \
\rho = \{ \rho_e \in {\bf loc}R_a \}_{e \in \Delta^a}
\ ,
\]
fulfill
\begin{equation}
\label{eq.C1}
\left\{
\begin{array}{ll}
dz(c)=1 \ \ , \ \ \forall c \in \Sigma_2(\Delta^a) \ ,
\\
z(b) \cdot \rho_{\partial_1b}(t) = \rho_{\partial_0b}(t) \cdot z(b)
\ \ , \ \
\forall b \in \Sigma_1(\Delta^a) \ , \ t \in R_a \ .
\end{array}
\right.
\end{equation}
The first equation says that $z$ is a cocycle defined on the "local" poset $\Delta^a$.
The second one yields an interpretation of $z$ as the flat connection
transporting from $\bl$ to $\bo$ the quantum charge represented by $\rho$ 
(see \cite[\S 25-28]{Rob} or the friendly \cite{Dop11}) 
and implies that $z$ uniquely determines $\rho_e$ on $R_e$ for any $e < a$:
for, if $o \in \Delta^a$ with $o \perp e$ and $p : e \to o$ is a path in $\Delta^a$ then
\begin{equation}
\label{eq.int0}
\rho_e(T') \ \stackrel{ (\ref{eq.C1}) }{=} \ z(p)^* \rho_o(T') z(p) \ , \ \forall T' \in R_a
\ \ \Rightarrow \ \
\rho_e(T) \ \stackrel{ (\ref{eq.C0}) }{=} \ z(p)^* T z(p) \ \ , \ \ \forall T \in R_e \ .
\end{equation}
The arrows of $Z^1_c(\efR)_{< a}$ are given by
\begin{equation}
\label{eq.int00}
(z_\rho,w_\sigma)
\ := \
\{ \ t = \{ t_e \in R_e \}_{e \in {\Delta^a}} 
     \ : \
     t_{\partial_0b} z(b)  =  w(b) t_{\partial_1b}
     \ , \
     \forall b \in \Sigma_1(\Delta^a) \ \}
\ .
\end{equation}
Note that $\rho,\sigma$ play no role in the definition of $(z_\rho,w_\sigma)$.
Applying repeatedly (\ref{eq.C1}) to a path
$p : e \to o$, $p = b_n * \ldots * b_1$, with $e \perp o < a$,
we have that if $t \in (z_\rho,w_\sigma)$ then
\[
t_o z_\rho(p) \ = \ w_\sigma(p) t_e
\ \ \Rightarrow \ \
T w_\sigma(p) t_e = T t_o z_\rho(p) \stackrel{Einstein \ causality}{=} t_o T z_\rho(p)
\ \ , \ \
\forall T \in R_e
\ ,
\]
so, applying (\ref{eq.int0}),
\begin{equation}
\label{eq.int}
\sigma_e(T) \, t_e \ = \ t_e \, \rho_e(T)
\ \ , \ \qquad
\forall t \in (z_\rho,w_\sigma)
\ , \
e < a
\ , \
T \in R_e
\ .
\end{equation}
We define the tensor product on $Z^1_c(\efR)_{< a}$. To this end, given 
$z_\rho , w_\sigma$, $z'_{\rho'} , w'_{\sigma'} \in$ ${\bf obj} Z^1_c(\efR)_{< a}$, 
$t \in (z_\rho,z'_{\rho'})$, 
$s \in (w_\sigma,w'_{\sigma'})$,
we set
\begin{equation}
\label{eq.tens.z1}
\left\{
\begin{array}{ll}
\{ z \otimes w \}_{\rho \circ \sigma}(b) := z_\rho(b) \rho_{\partial_1b}(w_\sigma(b))
\ \ , \ \
\forall  e < a
\ , \
b \in \Sigma_1(\Delta^a) \ ,
\\
(t \otimes s)_e := t_e \rho_e(s_e)
\ , \ 
\forall  e < a \ .
\\
\end{array}
\right.
\end{equation}
The fact that the above-defined relations yield a tensor product on $Z^1_c(\efR)_{< a}$
follows by elementary *-algebraic computations (analogous to those in \cite[\S 27]{Rob},
\cite[\S 4.2.1]{Ruz05}), anyway for convenience of non-specialists we recall some of them.
A first point is that since any $\rho_e$, $e < a$, is a unital *-endomorphism, we have that
$z_\rho(b) \rho_{\partial_1b}(w_\sigma(b))$ is unitary. Moreover, (\ref{eq.C0}) implies that
$z_\rho(b) \rho_{\partial_1b}(w_\sigma(b)) \in R_{|b|}$ for any $b \in \Sigma_1(\Delta^a)$. 
The cocycle relations follow by computing, for any $c \in \Sigma_2(\Delta^a)$,
\[
\begin{array}{ll}
z_\rho(\partial_0c) \rho_{\partial_{10}c}(w_\sigma(\partial_0c)) \cdot
z_\rho(\partial_2c) \rho_{\partial_{12}c}(w_\sigma(\partial_2c)) & =
z_\rho(\partial_0c) z_\rho(\partial_2c) 
\rho_{\partial_{12}c}(w_\sigma(\partial_0c)) \rho_{\partial_{12}c}(w_\sigma(\partial_2c)) = \\ & =
z_\rho(\partial_1c) \rho_{\partial_{12}c}(w_\sigma(\partial_1c)) = \\ & =
z_\rho(\partial_1c) \rho_{\partial_{11}c}(w_\sigma(\partial_1c)) 
\ ,
\end{array}
\]
having used repeatedly (\ref{eq.C1}),
the composition rule of face maps $\partial_{12} = \partial_{11}$ 
(see (\ref{eq.dd}) or \cite[Eq.1]{RR06}), and the cocycle relations of $z,w$. 
Passing to the "charges", for any $t \in R_a$ we have
\[
\begin{array}{ll}
z_\rho(b) \rho_{\partial_1b}(w_\sigma(b)) \cdot \{ \rho_{\partial_1b} \circ \sigma_{\partial_1b} \}(t)
& = \
z_\rho(b) \cdot \rho_{\partial_1b} ( w_\sigma(b) \sigma_{\partial_1b}(t) )
\ = \
\rho_{\partial_0b} ( \sigma_{\partial_0b}(t) w_\sigma(b) ) \cdot z_\rho(b)
\ = \\ & =
\{ \rho_{\partial_0b} \circ \sigma_{\partial_0b} \}(t) \cdot z_\rho(b) \rho_{\partial_1b}(w_\sigma(b)) \ ,
\end{array}
\]
so (\ref{eq.C1}(2)) is fulfilled for $z \otimes w$ and 
$(z \otimes w)_{\rho \circ \sigma} \in \obj Z^1_c(\efR)_{<a}$.
To verify that $t \otimes s$ is an arrow in
$( \ (z \otimes w)_{\rho \circ \sigma} \ , \ ( z' \otimes w')_{\rho' \circ \sigma'} \ )$
we compute
\[
\begin{array}{lcl}
t_\bo \rho_\bo(s_\bo) \cdot z_\rho(b) \rho_{\bl}(w_\sigma(b)) & = &
t_\bo z_\rho(b) \cdot \rho_\bl(s_\bo w_\sigma(b)) = \\ & = &
z'_{\rho'}(b) t_\bl \rho_\bl(w'_{\sigma'}(b)) \rho_\bl(s_\bl) = \\ & \stackrel{ (\ref{eq.int}) }{=} &
z'_{\rho'}(b) \rho'_\bl(w'_{\sigma'}(b)) \cdot t_\bl \rho_{\bl}(s_\bl)
\ ,
\end{array}
\]
obtaining the desired relations.

\begin{rem}
Let $\iota^a \in \obj Z^1_c(\efR)_{< a}$ be the pair defined by the identity cocycle of 
$Z^1(\efR)_{< a}$ and the constant family with element the identity $id_{R_a} \in {\bf aut}R_a$. 
Clearly $\iota^a$ is the identity object of the tensor product of $Z^1(\efR)_{< a}$, and, 
given $x \in a$, we have
\[
\bC 1                      \ \subseteq \
(\iota^a,\iota^a)          \ = \ 
\cap_{e < a} R_e           \ \subseteq \ 
\cap_{e < a : x \in e} R_e \ \stackrel{(*)}{=} \ 
\cap_{a \ni x} R_a         \ \stackrel{(\ref{eq.Ap1})}{=} \ 
\bC 1 \ ,
\]
where for the equality (*) we used the fact that $\{ a \in \Delta : x \in a \}$ is downward directed.
So $(\iota^a,\iota^a) = \bC 1$ and $Z^1_c(\efR)_{< a}$ has a simple unit.
\end{rem}

\subsection{Symmetry.}
\label{sec.sym}
As a final step we consider the symmetry structure. 
This is an essential ingredient at the  level of duality theory (\cite[\S 1]{DR89}) and, at the same time, 
encodes the Fermi-Bose statistics (see \cite{Dop11} for a friendly introduction). 
We introduce the family of unitary operators
\begin{equation}
\label{eq.sym.z1}
\eps (z_\rho,w_\sigma)_e := w_\sigma(p_{oe})^* \sigma_o(z_\rho(p_{o'e}))^* \cdot
                            z_\rho(p_{o'e}) \rho_e(w_\sigma(p_{oe}))
\ \ , \ \
\forall  e < a
\ ,
\end{equation}
where 
$p_{oe} : e \to o$, $p_{o'e} : e \to o'$ are paths in $\Delta^a$ such that $o \perp o'$.
The fact the for any $a \in \Delta$ there are causally disjoint $o,o' < a$ is an immediate consequence
of \cite[Lemma B.5]{BR08}, so any $\eps (z_\rho,w_\sigma)_e$, $e < a$, is well-defined as
an element of $UH \cap R_a$.
By Lemma \ref{lem.qft2} we have that 
$\eps (z_\rho,w_\sigma)_e$ is independent of the choice of $p_{oe}$, $p_{o'e}$, $o$ and $o'$.

\smallskip

The family $\eps (z_\rho,w_\sigma)$ fulfils the symmetry relations of \cite[\S 1]{DR89}, which, 
in our setting, take the form
\begin{equation}
\label{eq.sym}
\left\{
\begin{array}{ll}
\eps(z_\rho,w_\sigma)_\bo \cdot (z \otimes w)_{\rho \circ \sigma}(b) 
\ = \
(w \otimes z)_{\sigma \circ \rho}(b) \cdot \eps(z_\rho,w_\sigma)_\bl \ ,
\\
(s \otimes t)_e \cdot \eps(z_\rho,w_\sigma)_e 
\ = \
\eps(z'_{\rho'},w'_{\sigma'})_e \cdot (t \otimes s)_e \ ,
\\
\eps(z_\rho,w_\sigma)_e \cdot \eps(w_\sigma,z_\rho)_e \ = \ 1
\ \ \ , \ \ \
\eps(\iota^a,z_\rho)_e \ = \ \eps(z_\rho,\iota^a)_e = 1 \ ,
\\
\eps(  (z \otimes w)_{\rho \circ \sigma} , v_\tau  )_e
\ = \
\eps(z_\rho,v_\tau)_e \cdot \rho_e (\eps(w_\sigma,v_\tau)_e)
\ ,
\end{array}
\right.
\end{equation}
$\forall e < a$, 
$z_\rho,z'_{\rho'}, w_\sigma,w'_{\sigma'} ,v_\tau \in \obj Z^1_c(\efR)_{< a}$, 
$t \in (z_\rho,z'_{\rho'})$, $s \in (w_\sigma,w'_{\sigma'})$,
$b \in \Sigma_1(\Delta^a)$.

\smallskip

The explicit computations, that are on the line of \cite{Rob,Ruz05,BR08},
are postponed in \S \ref{app.sym}. It is at this point that we make use 
of the fact that we must have $s \geq 2$ spatial dimensions, as otherwise we may get a braiding instead of a symmetry. 
Since, as mentioned in the Introduction, where are interested in ordinary Bose-Fermi statistics,
in the rest of the section we shall assume that $M$ has spacetime dimension $\geq 3$.

\smallskip

Now, even if (\ref{eq.sym}) are fulfilled, in general 
$\eps (z_\rho,w_\sigma)$
\emph{does not} define a symmetry, because we should have
\[
\eps (z_\rho,w_\sigma)_e \in R_e \ \ , \ \ \forall e < a \ .
\]
The standard way to prove the above property would be the following.
In general, if $t \in (\rho_e,\sigma_e)$ then 
\[
t t' = t \rho_e(t') = \sigma_e(t') t = t't \ , \  \forall o \perp e , t' \in R_o
\ \Rightarrow \ 
t \in R_o' \ , \ \forall o \perp e \ ,
\]
and this, by (punctured) Haag duality, suffices to conclude that $t \in R_e$.
In particular, in our case we can verify that 
$\eps(z_\rho,w_\sigma)_e \in ( \sigma_e \rho_e , \rho_e \sigma_e )$,
nevertheless we are working with families of operators in the restricted net 
\[
\efR |_{\Delta^a} \ = \ (  \{ R_e \}_{e<a} \ , \jmath )_{\Delta^a} \ ,
\]
that could violate punctured Haag duality because, given $x \in M$, we expect that the inclusion
\[
R_e   \ =  \ \cap_{o \perp e , \ovl o \perp \{ x \} } R_o' 
\ \ \subseteq \ \
R^a_e \ := \ \cap_{o \perp e , o < a} R_o'
\]
is proper for any $e < a$. Thus we can conclude only that
$\eps(z_\rho,w_\sigma)_e \in R^a_e$, $e < a$.
In the next subsection we show how to avoid this difficulty using Rem.\ref{rem.back1}.

\subsection{The main result.}
We now define a presheaf structure on the family 
$Z^1_c(\efR)_< := \{ Z^1_c(\efR)_{<a} \}_{a \in \Delta}$,
by means of the restriction functors
\[
r_{aa'} : Z^1_c(\efR)_{<a'} \to Z^1_c(\efR)_{<a}
\ \ , \ \
( \ r_{aa'}z := z|_{\Sigma_1(\Delta^a)}
\ , \
r_{aa'}\rho := \{ \rho_e \}_{e \in \Delta^a} \ )
\ \ , \ \
r_{aa'}t := \{ t_e \}_{e \in \Delta^a}
\ ,
\]
which obviously fulfill the presheaf relations
$r_{aa''} = r_{aa'} \circ r_{a'a''}$, $\forall a \subseteq a' \leq a''$.

\begin{lem}
\label{lem.br1}
The pair $\efZ_\otimes := (Z^1_c(\efR)_<,r)^\Delta$ defines a simple tensor presheaf.
Moreover we have 
\begin{equation}
\label{eq.lem.br1}
\eps (z_\rho,w_\sigma)_e \ = \ \eps ( \ r_{aa'}z_\rho , r_{aa'}w_\sigma \ )_e
\ \ , \ \
\forall e < a \subseteq a'
\ ,
\end{equation}
where $z_\rho , w_\sigma \in \obj Z^1_c(\efR)_{<a'}$.
\end{lem}

\begin{proof}
We assume $a \subset a'$, otherwise there is nothing to prove.
Let $z_\rho,w_\sigma \in \obj Z^1_c(\efR)_{<a'}$ and $t \in (z_\rho,w_\sigma)$.
Then for any path $p : e' \to e$ such that $ e < a$, ${e'} < a'$, $|p| \subseteq \Delta^{a'}$, 
we have
\[
\| t_e \|^2 = 
\| z(p) t_{e'} w(p)^* \|^2 = 
\| w(p) t_{e'}^* t_{e'} w(p)^* \| = 
\| t_{e'}^* t_{e'} \| =
\| t_{e'} \|^2
\ ,
\]
so $r_{aa'}$ is isometric and $\efZ$ is a presheaf. 
In the following lines we prove that any $r_{aa'}$, $a \subseteq a'$, preserves tensor product and 
the families $\eps (\cdot,\cdot)$.
The fact that the tensor product is preserved is evident by (\ref{eq.tens.z1}).
Passing to the "symmetry", we consider $a \subseteq a'$, $z_\rho , w_\sigma \in Z^1_c(\efR)_{<a'}$
and note that the operators
\[
\eps ( \ r_{aa'}z_\rho , r_{aa'}w_\sigma \ )_e
\ \ , \ \
e \in \Delta^a
\ ,
\]
should be defined, according to (\ref{eq.sym.z1}), using paths 
$p_{oe}$, $p_{o'e}$ in $\Delta^a$ with $o \perp o'$;
but since $\Delta^a \subseteq \Delta^{a'}$ and $\Sigma_1(\Delta^a) \subseteq \Sigma_1(\Delta^{a'})$
we have that $p_{oe}$, $p_{o'e}$ are paths also in $\Delta^{a'}$, so, since any
$\eps (z_\rho,w_\sigma)_e$, $e \in \Delta^{a'}$,
is independent on the choice of $p_{oe}$, $p_{o'e}$, we obtain (\ref{eq.lem.br1}) as desired.
\end{proof}

\begin{lem}
\label{lem.br2}
$Z^1(\efR)$ is the category of sections of $\efZ_\otimes$.
\end{lem}

\begin{proof}
If $z \in \obj Z^1(\efR)$ then we define the family
$z^a := z|_{\Sigma_1(\Delta^a)}$, $\forall a \in \Delta$;
to add a family of charges, for any $e<a$ we set
\[
\left\{
\begin{array}{ll}
\Delta(\perp a,e) \ := \ \{ p \ : \ \partial_1p \perp a , \partial_0p = e \} \ ,
\\
\varrho^a_e(t) \ := \ \{ \ad z(p) \}(t) \ , \  \forall t \in R_a \ ,
\ \ {\mathrm{where}} \ \ p \in \Delta(\perp a,e) \ .
\end{array}
\right.
\]
By \cite[Eq.4.5]{BR08}, $\varrho^a_e$ does not depend on the choice of $p \in \Delta(\perp a,e)$,
and \cite[Lemma 5.1(i,iv)]{BR08} 
{\footnote{
In this reference an $x \in M$ is chosen with $a \perp \{ x \}$,
and we have the notation $\varrho^a_e \equiv y^z_x(e) |_{R_a}$, $\forall e \in \Delta^a$.
Moreover, it is proved that $\varrho^a_e$ is independent of $x$.
}}
imply that $\varrho^a_e \in {\bf loc}R_a$. Using \cite[Lemma 5.1(ii)]{BR08} we find
\begin{equation}
\label{eq.C1.0}
z(b) \cdot \varrho^a_{\partial_1b}(t) = \varrho^a_{\partial_0b}(t) \cdot z(b)
\ \ , \ \
\forall a \in \Delta \ , \ t \in R_a \ , \ b \in \Sigma_1(\Delta^a)
\ ,
\end{equation}
so
\[
z_\varrho^a \ := \ ( \ z^a \ , \ \varrho^a := \{ \varrho^a_e \}_{e<a} \ ) \ \in \ \obj Z^1_c(\efR)_{<a} \ .
\]
Let now $a \subseteq a'$; then $\Delta(\perp a',e) \subseteq \Delta(\perp a,e)$ so, 
by path-invariance of $\varrho^{a'}_e$,
\[
\varrho^{a'}_e |_{R_a} 
\ := \ 
\ad z(p')|_{R_a}
\ = \
\varrho^a_e
\ \ \ , \ \ \qquad
\forall p \in \Delta(\perp a',e) \subseteq \Delta(\perp a,e)
\ ,
\]
and we conclude that $\{ z_\varrho^a \}_a$ is a section of $\efZ$.
If $t \in (z,w)$ then defining 
$t^a_e := t_e$, $\forall e<a$,
yields an arrow $t^a$ of $Z^1_c(\efR)_{<a}$ and a section $\wt t := \{ t^a \}$, 
so $Z^1(\efR)$ is a subcategory of $\wt \efZ$.
On the other side, let us consider a section 
\[
(\wt z,\wt \varrho) \in \obj \wt \efZ
\ \ , \ \
(\wt z,\wt \varrho) \ = \ \{ (\wt z^a,\wt \varrho^a) \in \obj Z^1_c(\efR)_{<a} \}_{a \in \Delta} \ .
\]
Since any $a \in \Delta$ has proper closure in $M$, for any $b \in \Sigma_1(\Delta)$ there is $a \in \Delta$ such that 
$|b| \in \Sigma_1(\Delta^a)$ and we define
$z(b) := \wt z^a(b)$.
Since $\wt z$ is a section, the previous definition does not depend on the choice of $a > |b|$;
moreover given $c \in \Sigma_2(\Delta)$ we have $dz(c) = d \wt z^a(c) = 1$ for any $a > |c|$
so $z \in \obj Z^1(\efR)$. 
If $\wt t := \{ \wt t^a \}$ is an arrow of $\wt \efZ$ then for any $e \in \Delta$
there must be some $a \in \Delta$ with $e<a$ and we set 
$t_e := \wt t^a_e$,
so reasoning as above we have that $t := \{ t_e \}$ is an arrow of $Z^1(\efR)$.
We conclude that $Z^1(\efR) = \wt \efZ$ as desired.
\end{proof}

\begin{lem}
\label{lem.br3}
The subpresheaf $\efS := \efZ^\downarrow$, defined as in Rem.\ref{rem.back1}, 
is a simple, symmetric tensor presheaf and $\wt \efS = Z^1(\efR)$.
\end{lem}

\begin{proof}
By definition we have $\efS = (S,r)^\Delta$, where each $S_a$, $a \in \Delta$, is the full subcategory
of $Z^1_c(\efR)_{<a}$ with objects
\[
\obj S_a \ := \ \{  z_\varrho^a \ , \ z \in \obj Z^1(\efR) \} \ .
\]
By Rem.\ref{rem.back1} we already know that $\wt \efS = Z^1(\efR)$.
Moreover $\efS$ is clearly simple and tensorial, so we only have to check the symmetry property.
To this end we note that we already know from \S \ref{sec.sym} that, given $o,o' < a$, $o \perp o'$,
\[
\eps (z_\varrho^a , w_\varsigma^a)_e 
\ := \
w_\varsigma(p_{oe})^* \varsigma^a_o(z_\varrho(p_{o'e}))^* \cdot
z_\varrho(p_{o'e}) \varrho^a_e(w_\varsigma(p_{oe}))
\ \ , \ \
e < a
\ ,
\]
fulfils the symmetry relations, so we only have to prove that
$\eps (z_\varrho^a , w_\varsigma^a)_e \in R_e$ for any $e \in \Delta^a$.
Now, the point is that $\eps (z_\varrho^a , w_\varsigma^a)$ is defined
exactly as in \cite[Eq.5.10]{BR08} and this time, due to the fact that $z,w$ \emph{are globally defined},
the charges $\varrho^a_e , \varsigma^a_e$ extend to endomorphisms
$\varrho^x_e , \varsigma^x_e \in {\bf end}R^\perp_x$,
where $R^\perp_x := \cup_{ \ovl \omega \perp \{ x \} } R_\omega$ for any 
$x \in M$ such that $\{ x \} \perp \ovl e$ (see \cite[\S 5.1]{BR08}).
The property (\ref{eq.int}) now holds in the form
\[
\eps (z_\varrho^a , w_\varsigma^a)_e \cdot \{ \varrho^x_e \circ \varsigma^x_e \}(t)
\ = \
\{ \varsigma^x_e \circ \varrho^x_e \}(t) \cdot \eps (z_\varrho^a , w_\varsigma^a)_e
\ \ , \ \
\forall t \in R^\perp_x 
\ ,
\]
at varying of $x \in M$ with $\{ x \} \perp \ovl e$ (see \cite[Lemma 5.1(iii)]{BR08}).
So, as remarked in \S \ref{sec.sym}, 
\[
\eps (z_\varrho^a , w_\varsigma^a)_e \cdot t' \ = \ 
t' \cdot \eps (z_\varrho^a , w_\varsigma^a)_e 
\ \ , \ \
\forall t' \in R_\omega \subseteq R^\perp_x 
\ , \ 
\ovl \omega \perp \{ x \} \ , \ \omega \perp e
\ ,
\]
and we conclude,
by punctured Haag duality, that $\eps (z_\varrho^a , w_\varsigma^a)_e \in R_e$ as desired. 
\end{proof}

Before to state our main result we select the sectors having a conjugate.
For any $a \in \Delta$, we denote the full subcategory of $S_a$ with objects $z_\varrho^a$,
$z \in Z^{1,\bullet}(\efR)$, by $S^\bullet_a$, and define
\[
\efS^\bullet \ := \ (S^\bullet,r)^\Delta \ .
\]
By the previous Lemmata, and since by \cite[\S 5.3]{BR08} we have that $Z^{1,\bullet}(\efR)$ is a DR-category, we obtain:
\begin{thm}
\label{thm.br}
Let $\efR = (R,j)_\Delta \subseteq BH$ be a Haag-Kastler net where $\Delta$ is the base of diamonds 
of a globally hyperbolic spacetime $M$ with spacetime dimension $\geq 3$.
Then $\efS := (S,r)^\Delta$ is a simple, symmetric tensor presheaf such that 
$Z^1(\efR)$ is the category of sections of $\efS$.
In particular, $\efS^\bullet$ is a DR-presheaf with category of sections $Z^{1,\bullet}(\efR)$.
\end{thm}

In the following points we list some consequences of the previous theorem.

\paragraph{Holonomy of sectors.}
Let $z \in \obj Z^1(\efR)$ and ${}_\beta \wa z$ the presheaf bundle of tensor powers of $z$ defined in (\ref{eq.beta}).
By definition, the fibre ${}_\beta \wa z_a$, $a \in \Delta$, has objects the tensor powers 
$z^n_a := z \otimes \ldots \otimes z$, $n \in \bN$
(note that $a \in \Delta$ plays no role at the level of objects), and arrows
\[
{}_\beta(z^n_a , z^m_a) 
\ = \ 
\{  t = ( t_e \in R_e \, , \, e \in \cup_{a' \supseteq a} \Delta^{a'} ) 
    \, : \, 
    z^m(b) t_\bl = t_\bo z^n(b) \, , \, \forall b \in \Sigma_1(\cup_{a' \supseteq a} \Delta^{a'}) \}
\ .
\]
A simpler way to describe ${}_\beta(z^n_a , z^m_a)$ is the following: since any $a$ has compact closure, we can find $a' \in \Delta$
such that $a \in \Delta^{a'}$; thus for any $t \in {}_\beta(z^n_a , z^m_a)$ we can pick $t_a \in R_a$, and this suffices to determine $t$ 
because of the relation
\[
t_e \ = \ z^m(p) \, t_a \, z^n(p)^* \ \ , \ \ \forall e \in \cup_{a' \supseteq a} \Delta^{a'} \ ,
\]
where $p : a \to e$ is a path made of 1-simplices in $\Sigma_1(\cup_{a' \supseteq a} \Delta^{a'})$.
In this way, on the arrows, the restriction functors take the form
$r^z_{aa'}(t_{a'}) =  t_a$, $a \subseteq a'$,
and have inverse 
\[
r^z_{a'a}(t_a) \ = \ t_{a'} \ = \  z^m(b_{a'a}) \, t_a \, z^n(b_{a'a})^*
\ \ , \ \
b_{a'a} := (a',a;a') \in \Sigma_1(\Delta) \ .
\]
Thus, reasoning as in (\ref{eq.back3a}), we compute the holonomy of ${}_\beta \wa z$,
\begin{equation}
\label{eq.holz}
\chi'_z : \pi_1(\Delta) \to {\bf aut}{}_\beta \wa z_a
\ \ , \ \
\{ \chi'_z(p) \}(t_a) \ = \ z^m(p) \, t_a \, z^n(p)^*
\ \ , \ \qquad
p : a \to a 
\ , \
t_a \in {}_\beta(z^n_a , z^m_a) \ .
\end{equation}
Let now in particular $z \in \obj Z^{1,\bullet}(\efR)$.
By Theorem \ref{lem_CCS00}, the sector $z$ defines
a compact Lie group $G_z \subseteq \ud$ and a holonomy representation
$\chi_z : \pi_1(\Delta) \to QG_z := NG_z / G_z$,
which is a complete invariant of the presheaf bundle ${}_\beta \wa z$. 
In accord with the proof of \cite[Theorem 5.3]{Vas12}, the group $G_z$ is such that there is a symmetric tensor isomorphism
$F : {}_\beta \wa z_a \to \wa \pi_{G_z}$.
Regarding $QG_z$ as an automorphism group of $\wa \pi_{G_z}$ (see Ex.\ref{ex.back1}), by the argument used to prove \cite[Eq.5.5]{Vas12} we have 
$\chi_z \circ F = F \circ \chi'_z$.
It can be easily verified that $d$ is the statistical dimension of $z$ (see \cite[\S 5.3]{BR08}), 
but we prefer to not discuss here the details.

\paragraph{Dual objects and topologically trivial sectors.} 
By Theorem \ref{thm_DUAL}, any embedding 
$I : \efS^\bullet \to \efC$
yields a full embedding of symmetric tensor \sC categories 
$Z^{1,\bullet}(\efR) \hra {\bf Hilb}_\efG(\Delta)$,
where $\efG$ is a compact group net bundle, and induces lifts 
$\wa \chi_z : \pi_1(\Delta) \to NG_z$, $\chi_z = \wa \chi_z \, {\mathrm{mod}} G_z$,
of the holonomies defined by sectors $z \in \obj Z^{1,\bullet}(\efR)$.
Let now $Z^1_t(\efR)$ denote the subcategory of topologically trivial sectors, that is, 
of those $z \in \obj Z^1(\efR)$ such that there is a family $\psi = \{ \psi_a \in UH \}$ with
$z(b) = \psi_{\partial_0b}^* \psi_{\partial_1b}$, $b \in \Sigma_1(\Delta)$.
Then
$z^n(p) = 1$ for any $p : a \to a$ and $n \in \bN$,
and by (\ref{eq.holz}) we have that $\chi'_z(p)$ is the identity functor for any $z \in \obj Z^1_t(\efR)$.
Thus when $z \in \obj Z^{1,\bullet}_t(\efR)$ we have that $\chi'_z$ and $\chi_z$ are trivial, 
and there is at least one lift $\wa \chi_z$, the trivial one. 
We have a full embedding
$Z^{1,\bullet}_t(\efR) \hra {\bf Hilb}_{\efG_t}(\Delta)$,
where $\efG_t$ is the trivial net bundle with fibre the Doplicher-Roberts dual $G_t$ of $Z^{1,\bullet}_t(\efR)$. 
But, by the results in \cite{Vas12}, we may also have 
$Z^{1,\bullet}_t(\efR) \hra {\bf Hilb}_\efG(\Delta)$
for some nontrivial group net bundle $\efG$.


\paragraph{Cochains.} By Theorem \ref{thm.gerbe}, any sector $z \in \obj Z^{1,\bullet}(\efR)$ defines the 1--cochain
$u_z : \Sigma_1(\Delta) \to NG_z$
that classifies lifts $\wa \chi_z : \pi_1(\Delta) \to NG_z$, $q \circ \wa \chi_z = \chi_z$. 
The cochain $u_z$ defines a gerbe in the sense of \S \ref{sec.concl}, see (\ref{eq.CON1a}).

\paragraph{Characteristic classes.}
Let $P_z \to M$ denote the locally constant, principal $QG_z$-bundle defined by $\chi_z$. 
By (\ref{eq.ccs0}), we have characteristic classes 
$\zeta^\uparrow(z) := [\zeta^\uparrow(P_z)] \in H^{2k-1}(M,\bR / \bK)$
that vanish when $\chi_z$ is the trivial morphism, 
that is, when $z$ is a topologically trivial sector.
Thus we can measure the obstruction of $z$ being a topologically trivial sector in terms
of the odd cohomology of the given spacetime.

\paragraph{Chern classes of fields with "trivial gauge group".}
Let $\efR$ be a Haag-Kastler net such that $G_z$ is trivial for any $z \in \obj Z^{1,\bullet}(\efR)$.
Then $NG_z = QG_z = \ud$ and $\chi_z$ takes values in $\ud$. By (\ref{eq.ccs}) we have Chern classes
\begin{equation}
\label{eq.ccsz}
\left\{
\begin{array}{ll}
c_k(z) := c_k^\uparrow(\chi_z) \in H^{2k-1}(M,{\mathbb{R/Z}})
\ , \
\forall k=1,\ldots \ , d \ ,
\\
ccs(z) := d + \sum_{k=1}^d \frac{ (-1)^{k-1}  }{ (k-1)! } \ c_k(z)_{ {\mathrm{mod}}\bQ }
\ \in \bZ \oplus H^{odd}(M,{\mathbb{R/Q}}) \ .
\end{array}
\right.
\end{equation}
This is the situation that we expect to hold for quantum fields with trivial superselection structure 
in the Minkowski spacetime, as, for example, the free electromagnetic field. 
The classes (\ref{eq.ccsz}) are an interesting clue supporting the idea that Brunetti-Ruzzi sectors describe 
Aharonov-Bohm-type situations (\cite[\S 7]{BR08}): in fact, in the same way 
the monodromy phase of the quantum electromagnetic potential is described by a class in 
$H^1(M,{\mathbb{R/Z}})$, 
the holonomy of $z \in \obj Z^{1,\bullet}(\efR)$ is measured by the first Chern class $c_1(z) \in H^1(M,{\mathbb{R/Z}})$.
In this regard, we conjecture that the free electromagnetic field yields a net such that that $G_z$ is trivial for any 
$z \in \obj Z^{1,\bullet}(\efR)$.

\section{Group bundles as global gauge symmetries.}
\label{sec.gauge}

In the previous section we proved that the superselection structure $Z^{1,\bullet}(\efR)$ can be
described by a presheaf and, as a consequence of \cite{Vas12}, 
as the dual of a group (net) bundle $\efG$ when an embedding is given.
To illustrate the way in which $\efG$ should behave like a gauge group,
we give a class of examples in which an action of the type (\ref{eq.GA00}) is implemented and 
representations of the type defined in \S \ref{sec.back.nets} appear.
The procedure that we use is based on geometric arguments and can be applied in a very general setting
(the input is a field algebra carrying an action by a gauge group), thus it yields a new approach to 
construct superselection sectors with nontrivial holonomy, different from the one used in \cite{BFM09}
in the case of massive boson fields on the Einstein spacetime ($S^1 \times \bR$).

\subsection{Twisted field nets.}
Let $\efF = (F,\jmath)_\Delta$ be a net of $C^*$-algebras. We say that $\efF$ is $\bZ_2$-\emph{graded}
whenever there is a period $2$ automorphism
\[
\beta_a : F_a \to F_a
\ \ : \ \
\beta_{a'} \circ \jmath_{a'a} = \jmath_{a'a} \circ \beta_a
\ \ , \ \
\beta_a^2 = id_{F_a}
\ \ , \ \
\forall a \subseteq a' \in \Delta \ .
\]
In this case we set
\[
F^\pm_a \ := \ \{ t \in F_a : \beta_a(t) = \pm t  \}
\ \ , \ \
\forall a \in \Delta
\ ,
\]
so we have the usual projections
\[
F_a \to F^\pm_a
\ , \
t \mapsto t^\pm := 1/2 (t \pm \beta_a(t))
\ \ , \ \
\forall a \in \Delta \ .
\]
The following definition of normal commutation relations is a variant of \cite[Eq.3.6]{BR08}
taking account, differently from the above-cited reference, of the net structure of the involved representation.
At any rate, it is trivially verified that the following results remain valid with the older definition.
We say that a representation $(\pi,U)$ of $\efF$ is \emph{normal} whenever 
for any $a,o \in \Delta$ with $a \perp o$, $t_1 \in F_a$, $t_2 \in F_o$,
it turns out
\begin{equation}
\label{def.norm}
\left\{
\begin{array}{ll}
[ \pi_a(t_1^+) \, , \, \pi_o(t_2^+) ]^{net} =  
[ \pi_a(t_1^+) \, , \, \pi_o(t_2^-) ]^{net} = 
[ \pi_a(t_1^-) \, , \, \pi_o(t_2^+) ]^{net} = 0 \ ,
\\
\left[ \pi_a(t_1^-) \, , \, \pi_o(t_2^-) \right]^{net}_+ = 0 \ ,
\end{array}
\right.
\end{equation}
where the (anti)commutators are in the sense of (\ref{eq.nets1}) 
(with the convention $[ \cdot , \cdot ]^{net}_- \equiv [ \cdot , \cdot ]^{net}$).

Let now $H_*$ denote a Hilbert space.
A \emph{$\bZ_2$-graded group on $H_*$} is given by a group $G_*$ of unitary operators on $H_*$, 
metrizable and compact under the strong topology, and a central element 
$\gamma \in ZG_*$, $\gamma^2 = 1$,
where $1 \in G_*$ is the identity.
A \emph{$(G_*,\gamma)$-net of von Neumann algebras} is a net
\[
\efF_* = (F_*,\jmath_*)_\Delta \ \subset BH_*
\]
represented on the Hilbert space $H_*$ (that is, $\jmath_*$ is given by inclusion maps), 
such that
\begin{equation}
\label{eq.gauge1}
\alpha^g_a(t) := g t g^* \in F_{*,a}
\ \ , \ \
\forall g \in G_* \ , \ t \in F_{*,a} \ , \ a \in \Delta \ .
\end{equation}
We denote the fixed-point net by $\efR_* = (R_*,\jmath_*)_\Delta$, where
\[
R_{*,a} := F_{*,a} \cap G_*' \ \ , \ \ \forall a \in \Delta \ .
\]
Note that $\gamma$ induces a $\bZ_2$-grading on $\efF_*$. 
We say that $\efF_*$ is \emph{normal} whenever the defining Hilbert space representation $(id_{F_*},1)$ is normal in the sense of the previous lines. 
By (\ref{eq.nets1}), this agrees with the classical definition used, for example, in \cite{DR90}.
Finally, we define the normalizer group
\begin{equation}
\label{def.NG}
N_\gamma G_* \subset UH_* \ : \ V \in N_\gamma G_* \stackrel{\cdot}{\Leftrightarrow} 
\left\{  
\begin{array}{ll}
V F_{*,a} V^* = F_{*,a} \ , \ \forall a \in \Delta \ , \\
VgV^* \in G_* \ , \ \forall g \in G_* \ , \\
V \gamma = \gamma V \ .
\end{array}
\right.
\end{equation}
\begin{lem}[Twists of field nets]
\label{lem.gauge1}
Let $(G_*,\gamma)$ be a $\bZ_2$-graded group on the Hilbert space $H_*$ and 
$\efF_*$ a normal $(G_*,\gamma)$-net of von Neumann algebras on $H_*$.
Then for any morphism
\[
\chi : \pi_1(M) \to N_\gamma G_* \subset UH_*
\]
there are:
\begin{enumerate}
\item A group net bundle $\efG = (G,\hat \imath)_\Delta$ with fibre $G_*$;
\item A $\bZ_2$-graded net of \sC algebras $\efF = (F,\jmath)_\Delta$;
\item A gauge action $\efG \times_\Delta \efF \to \efF$ with fixed point net $\efR = (R,\jmath)_\Delta$;
\item A Hilbert net bundle $\efH = (H,U)_\Delta$ with holonomy $\chi$ and
      a normal representation $\pi : \efF \to \efB \efH$;
\item A causal
      {\footnote{That is $[ \pi_a(R_a) , \pi_o(R_o) ]^{net} = 0$ for any $o \perp a$, in the sense of (\ref{eq.nets1}).}}
      representation $(\pi|_R ,U|_R)$ of $\efR$.
\end{enumerate}
\end{lem}

\begin{proof}
\emph{Point 1.} Given $\omega \in \Delta$ and a path frame $P_\omega$ as in (\ref{def.pf}) we define 
\[
\imath_{a'a} := \chi( p_{\omega a'}*(a',a;a') * p_{a \omega})
\ \ , \ \qquad
\forall a \subseteq a' \ .
\]
By homotopy invariance the family $\imath$ fulfils 
$\imath_{a''a'} \imath_{a'a} = \imath_{a''a}$, $\forall a \subseteq a' \subseteq a''$
(see \cite[\S 2.2]{Ruz05}). We define
\[
G_a \equiv G_*
\ \ , \ \
\hat \imath_{a'a}(g) := \imath_{a'a} g \, \imath_{a'a}^*
\ \ , \ \qquad
\forall g \in G_* \ , \ a \subseteq a' \in \Delta \ .
\]
\emph{Point 2.} We define $F_a := F_{*,a}$, $\forall a \in \Delta$, and
\[
\jmath_{a'a} : F_a \to F_{a'}
\ \ , \ \
\jmath_{a'a}(t) := i_{a'a} t \, i_{a'a}^*
\ \ , \ \qquad
\forall a \subseteq a'
\ , \
t \in F_a
\ .
\]
Since $\imath_{a''a'} \imath_{a'a} = \imath_{a''a}$, $\forall a \subseteq a' \subseteq a''$,
we conclude that $\efF = (F,\jmath)_\Delta$ is a net. Moreover, since 
\[
\imath_{a'a} \gamma   \ = \ \gamma \, \imath_{a'a}
\ \ , \ \ 
\forall a \subseteq a' \ ,
\]
we have a $\bZ_2$-grading on $\efF$ given by the adjoint action of $\gamma$.
\emph{Point 3.} Taking into account that $G_a \equiv G_*$ and $F_a = F_{*,a}$ for any $a \in \Delta$
we define $\alpha^g_a \in {\bf aut}F_a$ as in (\ref{eq.gauge1}). 
To check that $\alpha$ is a gauge action (\ref{eq.GA00}) must be fulfilled, so we compute, for any $a \in \Delta$, $g \in G_a$, $t \in F_a$,
\[
\alpha^{\hat i_{a'a}(g)}_{a'} \circ \jmath_{a'a}(t) \ = \
\imath_{a'a} g \, \imath_{a'a}^* \cdot \imath_{a'a} t \, \imath_{a'a}^* \cdot \imath_{a'a} g^* \, \imath_{a'a}^* \ = \
\imath_{a'a} g t g^* \imath_{a'a}^* \ = \
\jmath_{a'a} \circ \alpha^g_a(t) \ ,
\]
and this yields the desired equalities. Note that $R_a = R_{*,a}$ for any $a \in \Delta$, because $G_* \equiv G_a$.
\emph{Point 4.} We set $H_a := H_*$, $\pi_a(t) := t$, and $U_{a'a} := \imath_{a'a}$ for any 
$a \subseteq a' \in \Delta$, $t \in F_a$. 
The pair $(\pi,U)$ is clearly a representation, that is normal since, for any $o \perp a$,
$t \in F_a$, $t' \in F_o$,
\[
\pi_a(t) \cdot U_p \pi_o(t') U_p \ = \ t \cdot \imath_p t' \imath_p^*
\ \ , \ \
\forall p : o \to a
\ ,
\]
and by construction $\imath_p t' \imath_p^* \in F_{*,o}$. 
Finally, by construction $(H,U)_\Delta$ has holonomy $\chi$.
\emph{Point 5.} Since $\efF$ is graded by $\gamma \in G_*$ we have $R_a \subseteq F^+_a$, $\forall a \in \Delta$,
so $\pi|_R$ is causal by applying (\ref{def.norm}).
\end{proof}

The following notion is a natural generalization of \cite[Def.3.4]{DR90}.
Let 
$\efF = (F,\jmath)_\Delta$ 
be a $\bZ_2$-graded net with a faithful normal representation
$(\pi,U)$
on the Hilbert net bundle $\efH =(H,U)_\Delta$, $H = \{ H_a \}$.
We assume that there is a net bundle of $\bZ_2$-graded compact groups acting on 
$\efH$,
that is, we have a family
$G = \{ G_a \subset UH_a \}$,
and central, period 2 elements $\gamma_a \in G_a$, $a \in \Delta$,
such that 
$\ad U_{a'a}(G_a) = G_{a'}$, $\ad U_{a'a}(\gamma_a) = \gamma_{a'}$, 
$a \subseteq a' \in \Delta$,
in such a way that the group net bundle
$\efG = (G, \ad U)_\Delta$
is defined. We assume that $\efG$ defines, by adjoint action, a gauge action on $\pi(\efF)$.
In this case, we call the quadruple 
$(\efF,\pi,U,\efG)$
a \emph{field system}.

\begin{defn}
We say that field systems $(\efF_1,\pi_1,U_1,\efG_1)$, $(\efF_2,\pi_2,U_2,\efG_2)$
are \textbf{equivalent} whenever there is a unitary family
$\nu = \{ \nu_a : H_{1,a} \to H_{2,a} \}$ 
such that
\begin{equation}
\label{eq.gauge2}
\nu_{a'} U_{1,a'a} = U_{2,a'a} \nu_a
\ \ , \ \
\ad \nu_{a'} \circ \jmath_{1,a'a} = \jmath_{2,a'a} \circ \ad \nu_a
\ \ , \ \
\nu_a g \nu_a^* \in G_{2,a} 
\ \ , \ \
\nu_a \gamma_{1,a} = \gamma_{2,a} \nu_a
\ ,
\end{equation}
for any $a \subseteq a'$, $g \in G_a$.
\end{defn}

Note that if $\chi_k$, $k=1,2$, are the holonomies of $(H_k,U_k)_\Delta$, 
then the first of (\ref{eq.gauge2}) implies
\begin{equation}
\label{eq.gauge3}
\chi_2(p) = \nu_a \cdot \chi_1(p) \cdot \nu_a^* \ \ , \ \ \forall p : a \to a \ .
\end{equation}

\begin{cor}
\label{cor.gauge1}
If the morphism $\chi$ of the previous Lemma takes values in 
$G_* \subseteq N_\gamma G_*$
then $\efR$ is isomorphic to $\efR_*$. 
Moreover, when $\chi$ is non-trivial $\efR_*$ is the fixed-point net of inequivalent field systems, 
namely 
$(\efF_*,id_{F_*},1,G_*)$ and $(\efF,\pi,U,\efG)$.
\end{cor}

\begin{proof}
Since $\chi$ takes values in $G_*$ we have $\imath_{a'a}(t) = t$ for any $a \subseteq a'$
and $t \in R_a$, so $\efR = \efR_*$. To prove that the above field systems are inequivalent 
it suffices to show that $U$ has a nontrivial holonomy, in fact the immersion 
$id_{F_*} : \efF_* \to BH$ has trivial holonomy and (\ref{eq.gauge3}) implies that
the holonomy of $U$ should be trivial too. But this is not the case, 
as the holonomy of $U$, by construction, coincides with $\chi$ (see the proof of Point 4 of the previous Lemma).
\end{proof}

One could argue that the only field net having a physical meaning is $\efF_*$.
But, as a matter of fact, the twisted field net $\efF$ has the same local algebras and causal properties as $\efF_*$,
and it can be constructed in such a way to support the Klein-Gordon and Dirac equations (see \S \ref{app.sym}).
Note that $\efF$ lives in a representation, $\pi$, having a nontrivial holonomy in the sense of \cite{BR08},
whose restriction to $\efR$ is equivalent to the (topologically trivial) representation of $\efR_*$ in the initial Hilbert space $H_*$;
thus the difference between $\efF$ and $\efF_*$ can be detected only in charged states which, in $\pi$,
carry topological effects given by the action of the fundamental group.

As a further remark, we observe that the technical ingredient underlying the construction of twisted field nets is that we allow 
$\jmath_{a'a}$, $a \subseteq a'$, 
to be actual *-monomorphisms and not simply inclusion maps. This is justified by the fact that the same happens for images
$\pi^z(\efR_*)$
of the observable net under the action of sectors $\pi^z$, $z \in \obj Z^1(\efR_*)$, see (\ref{eq.C00}).

\smallskip

At the conceptual level, the question in which we are interested in is whether $\efR_*$
and the superselection structure $Z^1(\efR_*)$ encode enough informations to select a unique field system.
On the grounds of \cite{Vas12} and the previous corollary we conclude that, to this end, further informations must be extracted from $Z^1(\efR_*)$.
A complete discussion on twisted field nets is beyond the purposes of the present paper, and we postpone it to a future work (see \S\ref{sec.concl}).

\subsection{Twisted nets and background potentials.}
\label{rem.qed}

We show how a class of twists arises when a quantum field interacts with a \emph{classical} background potential.
We start by assuming that $M$ is a globally hyperbolic spin manifold and consider the free Dirac field $\psi_*$
interacting with a (classical) closed 1--form $A \in Z^1_{deRham}(M)$.	

To this end we extract some informations from $A$.
A first point is that, since $A$ is closed, for any loop $\ell : [0,1] \to M$ the integral
$\int_\ell A $
is independent of the choice of $\ell$ in its homotopy class, thus we obtain the character
\begin{equation}
\label{eq,A0}
\zeta_A : \pi_1(M) \to \bT \ \ , \ \ \zeta_A([\ell]) \, := \, \exp 2\pi i \int_\ell A \ .
\end{equation}
Moreover, since any diamond $o \in \Delta$ is simply connected there is a local primitive $\phi_o \in C^\infty(o,\bR)$ 
such that $A|_o = d\phi_o$, and we have
\[
\phi_{o'} |_o - \phi_o \, = \, \Lambda_{o'o} \in \bR
\ \ \Rightarrow \ \ 
\Lambda_{o''o'} + \Lambda_{o'o} \, = \, \Lambda_{o''o}
\ , \ 
\forall o \subseteq o' \subseteq o'' \in \Delta
\ .
\]
Let now $b \in \Sigma_1(\Delta)$ and $\beta : [0,1] \to |b| \subset M$ such that $\beta(0) \in \partial_1b$, $\beta(1) \in \partial_0b$.
Then
\[
\int_\beta A \, = \, 
\int_\beta d\phi_{|b|} \, = \, 
\phi_{|b|}(\beta(1)) - \phi_{|b|}(\beta(0)) \, = \, 
\phi_{\partial_0b}(\beta(1)) - \phi_{\partial_1b}(\beta(0)) + \Lambda_b \ ,
\]
where 
$\Lambda_b \, := \, \Lambda_{|b|\partial_0b} - \Lambda_{|b|\partial_1b}$.
Note that the family $\Lambda = \{ \Lambda_b \}$ is a cocycle and it can be verified that it encodes the cohomology class of $A$ 
(see \cite[\S 3]{RRV08}). In particular the holonomy of $A$ can be reconstructed as follows.
Let $a \in \Delta$ and 
$p : a \to a$, $p = b_n * \ldots * b_1$; 
then using the isomorphism (\ref{eq.back1}) we have a homotopy class $[p] \in \pi_1(M)$ containing a loop 
$\ell : [0,1] \to M$, $\ell = \beta_n * \ldots * \beta_1$,
where each curve $\beta_k$, $k=1,\ldots,n$, fulfils the above properties with respect to $b_k \in \Sigma_1(\Delta)$ (see \cite[\S 2.5.1]{Ruz05}).
Thus, iterating the previous equality, we obtain 
\begin{equation}
\label{eq.A}
\int_\ell A \, = \, 
\sum_k \int_{\beta_k} A \, = \, 
\sum_k \Lambda_{b_k} 
\ .
\end{equation}
We now define a net by means of the family $\phi$ and the free Dirac field.
To this end we consider the spinor bundle $DM \to M$, the space $S^\infty_c(M,DM)$ of compactly supported, $C^\infty$-sections of $DM$ (spinors),
and, for any $o \in \Delta$, the space $\mS^\infty_o(M,DM)$ of spinors supported within $o \in \Delta$.
We note that we can define the product
\[
e^{2\pi i\phi_o} s \, \in S^\infty_o(M,DM)
\ \ , \ \ 
\forall o \in \Delta \ , \ s \in \mS^\infty_o(M,DM) \ .
\]
Let us consider the free Dirac field
$\psi_* : \mS^\infty_c(M,DM) \to BH_*$
in the sense of \cite{Dim82}. For any $o \in \Delta$ we define the operator-valued distribution
\[
\psi_o : \mS^\infty_o(M,DM) \to BH_*
\ \ , \ \ 
\psi_o(s) \, := \, \psi_*( e^{2\pi i\phi_o} s )
\ .
\]
Clearly, defining the field $\psi_o$ corresponds to considering the local gauge transformation 
\[
\psi_*(x) \, \to \, e^{2\pi i\phi_o(x)}\psi_*(x) \ \ , \ \ x \in o \ ,
\]
which yields in $o$ a solution for the Dirac equation with interaction given by $A$.
The net of von Neumann algebras $\efF_* = (F,\jmath_*) \subset BH_*$,
\[
F_o \, := \, \{ \psi_o(s) \, : \, s \in \mS^\infty_o(M,DM) \}'' \subset BH_* \ \ , \ \ o \in \Delta \ ,
\]
coincides with the one defined by $\psi_*$, 
thus at a first stage $\efF_*$ does not seem to preserve informations on the potential $A$;
but, in the following lines, we shall show that this is not the case.

We consider global gauge transformations
$\alpha_z(\psi_o(s))  :=  z \psi_o(s)$, $s \in \mS^\infty_o(M,DM)$, $z \in \bT$,
that are unitarily implemented in the fermionic Fock space $H_*$ by universality of the CARs,
\begin{equation}
\label{eq.gdirac}
\alpha_z(\psi_o(s)) \, = \, V_z \psi_o(s) V_z^* \ \ , \ \ \forall z \in \bT \ ,
\end{equation}
and set $G_* := \{ V_z , z \in \bT \} \subset UH_*$ with grading element $\gamma = V_{-1}$.
This yields the observable subnet 
$\efR_* = (R,\jmath_*)_\Delta$, $R_o = F_o \cap G_*'$, $o \in \Delta$.

\smallskip

Now, for any $o \subseteq o'$ we consider the map
\[
\jmath_{o'o}(\psi_o(s)) \, := \, 
\psi_{o'}(s) \, = \, 
\psi_*( e^{2\pi i (\phi_o + \Lambda_{o'o}) }s ) \, = \, 
e^{2\pi i \Lambda_{o'o}} \psi_o(s)
\ \ , \ \
s \in \mS^\infty_o(M,DM)
\ .
\]
By universality, any $\jmath_{o'o}$ extends to a *-monomorphism $\jmath_{o'o} : F_o \to F_{o'}$ 
and we obtain a net $\efF = (F,\jmath)_\Delta$.
A natural representation of $\efF$ can be constructed as follows:
we define the Hilbert net bundle $\efH = (H,U)_\Delta$, 
\[
H_o \equiv H_* \ \ , \ \ U_{o'o} \, := \, V_{ \exp 2\pi i \Lambda _{o'o} } \ \ , \ \ \forall o \subseteq o' \ ,
\]
and then note that, by (\ref{eq.gdirac}),
\[
\jmath_{o'o}(t) \, = \, 
\alpha_{ \exp 2\pi i \Lambda _{o'o} }(t) \, = \,
U_{o'o} \, t \, U_{o'o}^* 
\ \ , \ \ 
\forall o \subseteq o' \ , \ t \in F_o \ .
\]
The holonomy of $\efH$ is easily computed using (\ref{eq,A0}) and (\ref{eq.A}):
\[
U_p \, = \, \prod_k V_{\exp 2\pi i \Lambda _{b_k}} \, = \, V_{\zeta_A([\ell])} \in G_*
\ \ , \ \ 
\ell \in [p] \in \pi_1(M)
\ .
\]
Thus $\efF$ is a twisted net faithfully represented in the Hilbert net bundle $\efH = (H,U)_\Delta$. 
Performing a parallel displacement around the loop $p : a \to a$ yields a phase factor affecting the charged field operators,
\begin{equation}
\label{eq.AB2}
\ad U_p(\psi_a(s)) \, = \, \zeta_A([\ell]) \, \psi_a(s) \, = \, \exp 2\pi i \int_\ell A \, \cdot \, \psi_a(s)
\ \ , \ \
\ell \in [p] \in \pi_1(M)
\ .
\end{equation}
Thus, if $v_0 \in H_*$ is the vacuum, then the relative phase $\exp 2\pi i \int_\ell A$ appears between the states 
$\psi_a(s)v_0$, $U_p \psi_a(s) v_0 \in H_*$ (we use the notation in which $\psi_a(s)$ is defined by the creation operator).

\smallskip

To give an example we parametrize $\bR^4$ by the coordinates $x_0 , x_1 , x_2 , x_3$ and define $M = \bR^4 - S_\delta$,
where $S_\delta$ is the closed tube centered around the axis $x_3$ and having ray $\delta > 0$.
The base $\Delta$ is given by the one of double cones contained in $M$, which generates a nontrivial topology in the sense that
$\pi_1(\Delta) \simeq \pi_1(M) \simeq \bZ$
(note that $M$ is homotopic to $S^1$).

We consider a vector potential $A' \in \Omega^1(\bR^4)$ such that $F := dA' \in \Omega^2(\bR^4)$ is supported within $S_\delta$.
%
%
Note that this scenario is modeled in such a way to obtain a situation analogous to the one of the Aharonov-Bohm effect, see \cite[\S 15.5]{Fey}.
We can always find such 1--forms $A'$ by starting from $F \in \Omega^2(\bR^4)$ fulfilling the above property and defining $A'$ as a primitive of $F$.

As a final step we define $A \in \Omega^1(M)$, $A := A'|_M$. In this way, the previous hypothesis say that $A$ is closed and we can apply the 
previous construction, that gives rise to the phase factor (\ref{eq.AB2}). This last coincides with the phase affecting the wavefunctions of electrically charged particles
in the Aharonov-Bohm effect (see \cite[Eq.15.33]{Fey}).

\subsection{Current algebras on $S^1$.}
\label{sec.BMT}

We now apply Lemma \ref{lem.gauge1} to current algebras on $S^1$ defined as in the approach of Buchholz-Mach-Todorov
(\cite{BMT1,BMT2}); here, at a first stage, fields fulfilling braided commutation relations appear, anyway after a "bosonization" procedure a causal field net is constructed.
For a complete discussion of the model we refer the reader to the above-cited references;
here we limit ourselves to give the basic definitions, needed to see how Lemma \ref{lem.gauge1} applies.

\smallskip

The base $\Delta$ of $S^1$ is the one of intervals with proper closure, 
and we write $o \perp w \Leftrightarrow o \cap e = \emptyset$.
We take a fixed $n \in \bN$ (with no squares in the factor decomposition, see the above-cited references),
parametrize $S^1$ with complex numbers of modulus $1$, and for every $o \in \Delta$ consider the space
\begin{equation}
\label{eq.BMT0}
\mV^n_o \, := \, 
\{ \, q \in C^\infty(S^1,\bC) \, : \, {\mathrm{supp}}(q) \subseteq o \, , \, 
                                   zq(z) \in \bR , \forall z \in S^1 \, , \, 
                                   \oint \frac{1}{2\pi i} \, q(z) \, dz \, = \sqrt{2n} \, \}
\, .
\end{equation}
Here every $q \in \mV^n_o$ is interpreted as a charge distribution supported within $o \in \Delta$, see \cite[Eq.2.24]{BMT1}.
A field net $\efF^n_* = (F^n,\jmath)_\Delta \subset BH_*$ is defined, with local 
von Neumann algebras
\[
F^n_o \, = \, \{  \phi_q \in UH_* \, : \, q  \in \mV^n_o \}'' \ \ , \ \ o \in \Delta \ ,
\]
see \cite[Eq.4.93]{BMT1} and the following remarks.
The net $\efF^n_*$ is bosonic in the sense that $[ \phi_q , \phi_{q'} ] = 0$ whenever $\phi_q \in F_o$, $\phi_{q'} \in F_e$ with $o \perp e$,
see the discussion after \cite[Eq.4.82]{BMT1}.
There is a gauge group $G_* \simeq \bZ_{2n}$ generated by a unitary 
\[
V \, = \, e^{2\pi i (2n)^{-1/2} Q} \in UH_* \ \ : \ \ V^{2n} = 1 \ ,
\]
on which we define the grading element $\gamma = 1$; 
in the previous expression $Q$ denotes an unbounded operator on $H_*$ affiliated with the von Neumann algebra generated by $F^n$,
see \cite[Eq.4.40]{BMT1}, which is interpreted as the charge. 
We have
\begin{equation}
\label{eq.BMT1}
V \phi_q V^* \, = \, e^{i\pi / n} \phi_q
\ \ , \ \ 
\forall o \in \Delta \ , \ q \in \mV^n_o
\ ,
\end{equation}
see \cite[Eq.4.94]{BMT1}.
We can apply Lemma \ref{lem.gauge1} with morphisms of the form
\[
\chi : \bZ \simeq \pi_1(S^1) \to G_* \ ,
\]
that are determined by the choice of $\chi(1)$ that must be of the type $\chi(1) = V^\kappa$, $\kappa \in \{ 1 , \ldots , 2n-1 \}$.
As a consequence of Cor.\ref{cor.gauge1}, the twisted net $\efF$ defined by $\chi$ has observable net isomorphic to the untwisted net $\efR^n_* = (R^n,\jmath)_\Delta$
generated by operators in $\efF^n_*$ commuting with $V$.
We note that $V$ is, indeed, a central element of the von Neumann algebra $\vec{R}^n \subset BH_*$ generated by the family $R^n$, see \cite[Eq.4.90]{BMT1}. Any $\phi_q$ defines a localized automorphism of $\vec{R}^n$,
\[
\gamma_q \, := \, \ad \phi_q 
\ \ , \ \
\gamma_q(t) = t \, , \, \forall t \in R_e \, , \, e \perp {\mathrm{supp}}(q) \ .
\]
Passing to the twisted net $\efF^n$ we have a representation $\pi : \efF^n \to \efB\efH$, where $\efH = (H,U)_\Delta$ is the Hilbert net bundle
with fibres $H_o \equiv H_*$, $o \in \Delta$, defined as in Lemma \ref{lem.gauge1},
\[
\pi_o : F^n_o \to BH_* \ \ , \ \ \ad U_{o'o} \circ \pi_o \,= \, \pi_{o'}|_{F^n_o} \ \ , \ \ \forall o \subseteq o' \ ;
\]
here any $\pi_o$ is the inclusion map $F^n_o \to BH_*$, thus in the sequel it will be dropped from the notation.
It is instructive to exhibit the effect of the parallel displacement along a loop on the charged fields.
To this end, to fix the ideas we take $\kappa = 1$, that means $\chi(1) = V$.
Given $a \in \Delta$, the homotopy class of any loop $p : a \to a$ is labelled by an integer $[p] \in \bZ$, thus
\[
\ad U_p(\phi_q) \, = \, V^{[p]} \phi_q V^{[p],*} \, = \, e^{i\pi [p] / n} \phi_q 
\ \ , \ \ 
\forall \phi_q \in F^n_a
\ ,
\]
and a relative phase appears between the charged field and its parallel displacement.
We note that the interpretation of $\ad U_p(\phi_q)$ as the displacement of the field operator $\phi_q$ inducing 
the DHR automorphism $\gamma_q$ is in accord with the standard formulation of charge transporters, see \cite[\S 3.4.5]{Rob0}.

\subsection{Existence of field nets with a given gauge symmetry.}
\label{sec.free}
The following construction is on the line \cite{DP02}, where similar results are proven for the Minkowski spacetime.
We prove that for every suitable compact group $G_0$, globally hyperbolic spacetime $M$ and flat bundle $\mG \to M$ with suitable structure group, 
there is a twisted field net over which $\mG$ acts as a group of internal symmetries.
The "untwisted" net from which we start is defined by means of generalized free fields and yields a large class of toy models where the
analysis of topologically trivial superselection sectors of \cite{GLRV01,Ruz05} can be checked, as happens for \cite{DP02} in the Minkowski spacetime.

\smallskip

In the sequel we shall consider a $\bZ_2$-graded group $(G_0,\gamma)$, not necessarily realized on a Hilbert space,
where $G_0$ is compact and metrizable.
Let $\lambda$ be a symmetric and separating set
{\footnote{With \emph{symmetric} we mean that if $\varrho \in \lambda$ then the same is true for
           the conjugate representation $\ovl \varrho$, whilst the term \emph{separating} indicates
           that for any $g \neq g' \in G_0$ there is $\varrho \in \lambda$ such that $\varrho(g) \neq \varrho(g')$.}}
of irreducible subrepresentations of the regular representation of $G_0$.
We consider the corresponding Hilbert space
\[
L_\lambda := \oplus_{\varrho \in \lambda} L_\varrho \ \subset L^2(G_0) 
\]
on which $G_0$ acts by means of the faithful unitary representation
\begin{equation}
\label{def.rep}
G_0 \to UL_\lambda
\ \ , \ \
g \mapsto g_\lambda := \oplus_\varrho g_\varrho
\ ,
\end{equation}
where any $g_\varrho$ is a unitary acting on the irreducible, finite dimensional $G_0$-Hilbert space $L_\varrho$.
The grading element $\gamma$ induces the spectral decomposition
\[
L_\lambda = L_\lambda^+ \oplus L_\lambda^- \ ,
\]
which decomposes (\ref{def.rep}) into
\begin{equation}
\label{def.rep.pm}
G_0 \to UL_\lambda^\pm
\ \ , \ \
g \mapsto g_\lambda^\pm := \oplus_{\varrho(\gamma) = \pm 1} g_\varrho \ .
\end{equation}
Finally, we consider a real structure (that is, an antilinear, isometric map)
$J : L_\lambda^+ \to L_\lambda^+$, $J^2 = 1$,
commuting with (\ref{def.rep}) and stabilizing the subspaces $L_\varrho$ 
for any $\varrho \in \lambda$ with $\varrho(\gamma) = 1$, so that we have real structures
$J_\varrho : L_\varrho \to L_\varrho$.
We denote the \emph{real} subspace of $J$-invariant vectors in $L_\lambda^+$ 
by $L_\lambda^{+,J}$, which is a real Hilbert space with the scalar product induced by $L_\lambda^+$.

Now, we say that a compactly supported $f : M \to L_\lambda^{+,J}$ is $C^\infty$
whenever the function
$f_v(x) := ( f(x),v ) \in \bR$, $x \in M$,
is $C^\infty$ for any $v \in L_\lambda^{+,J}$, and denote the corresponding real vector space
by $C^\infty_c(M,L_\lambda^{+,J})$. 
Given the decomposition 
$f(x) = \{ f_\varrho(x) \in L_\varrho \}$, $x \in M$,
we define the real vector space
\[
\mD^+ \ := \ 
\{ f \in C^\infty_c(M,L_\lambda^{+,J}) \ : \
   f_\varrho \neq 0 \ {\mathrm{for \ finite \ }} \varrho \in \lambda  \} \ .
\]
The representation (\ref{def.rep}) induces an action by invertible linear maps,
\begin{equation}
\label{eq.CON0a}
G_0 \to GL \mD^+
\ \ : \ \
f_g(x) := g_\lambda^+ f(x)
\ \ , \ \qquad
\forall f \in \mD^+ \ , \ g \in G_0 \ , \ x \in M \ .
\end{equation}
Now, let us suppose that $M$ has a spin structure 
(this is guaranteed for 4-dimensional globally hyperbolic spacetimes, see \cite{Ish78} and related references).
We consider the spinor bundle $DM \to M$ (with fibre $\bC^4$), whose tensor product by $L_\lambda^-$
yields the bundle of Hilbert spaces 
$D_\lambda^- M \to M$, $D_\lambda^- M := L_\lambda^- \otimes DM$.
We have the direct sum decomposition 
\[
D_\lambda^- M \ \simeq \ \oplus_{\varrho(\gamma)=-1} D_\varrho M 
\ ,
\]
where any $D_\varrho M := L_\varrho \otimes DM \to M$ is a \emph{smooth} vector bundle.
Any section $s : M \to D_\lambda^- M$ decomposes as
$s = \{ s_\varrho : M \to D_\varrho M \}$,
and we say that $s$ is $C^\infty$ whenever every $s_\varrho$ is a $C^\infty$ section.
We denote the set of compactly supported, $C^\infty$ sections of $D_\lambda^-M$ by $\mS^\infty_c(M,D_\lambda^-M)$
and define
\[
\mD^- \ := \ 
\{ s \in \mS^\infty_c(M,D_\lambda^-M) \ : \
   s_\varrho \neq 0 \ {\mathrm{for \ finite \ }} \varrho \in \lambda \} \ ,
\]
that is a complex vector space with action
\begin{equation}
\label{eq.CON0b}
G_0 \to GL \mD^-
\ \ : \ \
s_g(x) := \{ g_\lambda^- \otimes 1 \} s(x)
\ \ , \ \
\forall s \in \mD^- \ , \ g \in G_0 \ , \ x \in M \ .
\end{equation}

\begin{lem}[Existence of field nets on hyperbolic spacetimes]
\label{lem.gauge2}
For any globally hyperbolic (spin) manifold $M$ and $\bZ_2$-graded group $(G_0,\gamma)$,
there are a Hilbert space $H_*$ with a faithful, strongly continuous representation 
$\tau : G_0 \to UH_*$,
and a normal $( \tau(G_0),\tau(\gamma) )$-net of von Neumann algebras $\efF_*$ on $H_*$.
\end{lem}

The proof of the previous Lemma is postponed to \S \ref{app.sym} and is based on the construction
of a symplectic form on $\mD^+$ and a scalar product on $\mD^-$, that allow us to define, respectively, 
\begin{itemize}
\item a selfadjoint Klein-Gordon field $\phi(f)$, $f \in \mD^+$, on the Hilbert space $H_*^\phi$,
\item a Dirac field $\psi(s)$, $s \in \mD^-$, on the Hilbert space $H_*^\psi$,
\end{itemize}
and then define $H_* := H_*^\phi \otimes H_*^\psi$ and $F_{*,a}$, $a \in \Delta$, as the von Neumann algebra
generated by 
\[
\{ e^{i \phi(f)} \otimes 1_\psi \ , \ 1_\phi \otimes \psi(s) \ : \ 
   f \in \mD^+ \ , \ s \in \mD^- \ , \ {\mathrm{supp}} \, (f) \, , \, {\mathrm{supp}} \, (s) \subseteq a \} \ .
\]
We now define the (compact) normalizer group
\[
N^\lambda G_0 \ := \ \{ V = \oplus_\varrho V_\varrho \in UL_\lambda \ : \ 
                        V_\varrho g_\varrho V_\varrho^* \in \varrho(G_0) \ , \ \forall g \in G_0  \} \ .
\]
Note that $N^\lambda G_0$ acts by adjoint action on $G_0$ and, since $\gamma_\varrho = \pm 1$ for any $\varrho \in \lambda$,
we have
\[
V \gamma_\lambda V^* = \gamma_\lambda \ \ , \ \ \forall V \in N^\lambda G_0 \ .
\]
Moreover, any $V \in N^\lambda G_0$ splits as $V = V^+ \oplus V^-$, $V^\pm \in UL_\lambda^\pm$.
\begin{thm}
\label{thm.gauge1}
Let $M$ denote a globally hyperbolic (spin) manifold and $(G_0,\gamma)$ a $\bZ_2$-graded group. 
Given a symmetric, separating set $\lambda$ for $G_0$ and a morphism
\[
\chi : \pi_1(M) \to N^\lambda G_0 \ ,
\]
there is a $\bZ_2$-graded net of \sC algebras $\efF = (F,\jmath)_\Delta$ endowed with:
\begin{enumerate}
\item A Hilbert net bundle $(H,U)_\Delta$ carrying a normal representation $(\pi,U)$ of $\efF$;
\item A gauge action $\efG \times_\Delta \pi(\efF) \to \pi(\efF)$ 
      with fixed-point net $\efR$ causally represented on $(H,U)_\Delta$,
      where $\efG = (G, \hat \imath)_\Delta$ is the group net bundle with fibre $G_0$ and holonomy $\ad \chi$.
\end{enumerate}
\end{thm}

\begin{proof}
With the notation of the previous Lemma we define $G_* := \tau(G_0) \subset UH_*$.
So $(G_*,\tau(\gamma))$ is a $\bZ_2$-group on $H_*$, and by the previous Lemma we have a 
normal $(G_*,\tau(\gamma))$-net of von Neumann algebras $\efF_*$ on $H_*$.
Now, by Lemma \ref{lem.gauge1}, to prove the Theorem it suffices to show that the unitary representation
$\tau$ of the previous Lemma extends to a unitary representation $\tau$ of $N^\lambda G_0$ taking values
in $N_{\tau(\gamma)} G_*$, in fact in such a case we would have the morphism 
$\tau \circ \chi : \pi_1(M) \to N_{\tau(\gamma)} G_*$.
But the construction of the desired representation is easy, in fact it suffices to define the actions 
\[
\left\{
\begin{array}{ll}
f_V(x) := V^+ f(x)
\ , \
\forall x \in M
\ , \
f \in \mD^+
\\
s_V(x) := \{ V^- \otimes 1 \} s(x)
\ , \
\forall x \in M
\ , \
s \in \mD^-
\end{array}
\right.
\ \ , \ \
V \in N^\lambda G_0 \ ,
\]
that yield, by universality of the CCRs and CARs, the desired extension $\tau$ fulfilling
\[
\tau_V \{ e^{i \phi(f)} \otimes 1_\psi \} \tau_V^* \ = \ e^{i \phi(f_V)} \otimes 1_\psi 
\ \ , \ \
\tau_V \{ 1_\phi \otimes \psi(s) \} \tau_V^* \ = \  1_\phi \otimes \psi(s_V)
\ .
\]
The verifications that $\tau(N^\lambda G_0) \subseteq N_{\tau(\gamma)} G_*$
now are trivial (see (\ref{def.NG})) and the Theorem is proved.
\end{proof}

\section{Conclusions and the gerbe perspective.}
\label{sec.concl}

In the present paper we proved that superselection structures in curved spacetimes are section categories of 
presheaves of symmetric tensor \sC categories.
As a consequence, several new invariants are assigned to a sector:
a holonomy with values in a compact Lie group, characteristic classes, a 1--cochain defining a non-abelian cocycle.
As a further novelty, when an embedding is given a Tannaka duality describes the superselection structure in terms of 
equivariant Hilbert net bundles (or equivariant flat Hermitian bundles, see \cite{RRV08,Vas12}),
and this interpretation is supported by constructions of nets generated by quantum fields acted upon by group bundles.
These constructions (twisted field nets) are very general and can be applied to any (free or interacting) field
generating a Haag-Kastler net fulfilling normal commutation relations; 
moreover, they yield the expected phase factors on the charged fields that, in particular, can arise
as holonomies of classical, interacting potentials in accord with the scenario of the Aharonov-Bohm effect (\S \ref{rem.qed}).

Two points, that are object of a work in progress, remain to be discussed.

The first one is a further discussion of nets of the type studied in \S \ref{sec.gauge} from the viewpoint of sectors: 
apart from the verification of all the properties of Haag-Kastler net (punctured Haag duality, first of all),
we are interested to determine which representations, among those constructed in Theorem \ref{thm.gauge1}, 
fulfill the selection criterion \cite[\S 3.2]{BR08},
and then to prove that these are in one-to-one correspondence with representations of $\efG$,
on the line of \cite[Theorem 3.14-3.15]{GLRV01}.

The second point concerns the reconstruction, starting from the Haag-Kastler net $\efR$ and 
the presheaf $\efS$, of the \sC gerbe playing the role of the field net \cite{DR90}.
The program is to construct the local field algebras $F_a$, $a \in \Delta$, as the crossed products
\[
F_a \, := \, R_a \rtimes S^\bullet_a 
\]
in the sense of \cite[Theorem 4.14]{BL}, that can be applied because the DR-category (with simple units) $S^\bullet_a$ embeds in ${\bf end}R_a$
(which has a non-simple unit because the centre of $R_a$ is not trivial).
The further step is to recognize that there is a unique group gerbe encoding all the eventual precosheaf structures on the family $F$,
and this will follow by using universality of the above crossed product.
In this scenario, the inequivalent field systems of Cor.\ref{cor.gauge1} defining the same observable net
would find their justification, given by the fact that the involved group net bundles are different manifestations of the same gerbe;
but even if the group net bundles have isomorphic duals, 
only one "physical" field net $\efF$ should be selected by a proper set of invariants of $Z^{1,\bullet}(\efR)$.
The net $\efF$ is expected to contain a subnet $\efF_t$, the one associated with the subcategory $Z^{1,\bullet}_t(\efR)$ of $Z^{1,\bullet}(\efR)$.

\smallskip

On a similar research line, we would like to mention the paper \cite{Few12}, where a gauge group is defined in a functorial way. 
The scenario is the one of locally covariant quantum field theory (\cite{BFV03}), where generic functors are considered instead of nets.
In our setting, it seems natural to conjecture that Fewster's gauge group is the Doplicher-Roberts dual of $Z^{1,\bullet}(\efR)$.

\smallskip

In the following paragraph we explain what we mean by a gerbe over a poset
and give a construction of gerbes of \sC algebras by means of quantum fields.

\paragraph{Gerbes.}
We recall the definition of the sets
\[
N_1(\Delta) := \{ \, b = ( b_0 \subseteq |b| \in \Delta ) \, \}
\ \ , \ \
N_2(\Delta) := \{ \, c = ( c_0 \subseteq c_1 \subseteq |c| \in \Delta ) \, \} 
\ ;
\]
by \cite[\S 2.2]{RRV07}, there are inclusions 
$N_1(\Delta) \subseteq \Sigma_1(\Delta)$, $N_2(\Delta) \subseteq \Sigma_2(\Delta)$.
Let now $G$ be a group. A {\em $G$-gerbe} over $\Delta$ is a pair
$\check \efG = (i,\delta)_\Delta$,
\begin{equation}
\label{eq.CON1}
i : \Sigma_1(\Delta) \to {\bf aut}G
\ \ , \ \
\delta : \Sigma_2(\Delta) \to G
\ \ : \ \
\ad \delta_c \circ i_{\partial_1c} = i_{\partial_0c} \circ i_{\partial_2c}
\ , \
\forall c \in \Sigma_2(\Delta)
\ .
\end{equation}
The basic idea is that $i$ is a ${\bf aut}G$-cocycle (in the sense of \cite{RRV07}) 
only up to inner automorphisms, and the obstacle to get the cocycle relations is encoded by $\delta$. 
If $\efR$ is a Haag-Kastler net, then any sector $z \in \obj Z^{1,\bullet}(\efR)$ defines a gerbe $\check \efG_z=(i_z,\delta_z)_\Delta$, 
by means of the cochain $u_z : \Sigma_1(\Delta) \to NG_z$ defined by Theorem \ref{thm.gerbe}:
\begin{equation}
\label{eq.CON1a}
i_z(b) := \ad u_z(b) \ \ , \ \ \delta_z(c) := du_z(c)
\ \ , \ \qquad
b \in \Sigma_1(\Delta) \ , \ c \in \Sigma_2(\Delta) \ .
\end{equation}
A {\em $\check \efG$-\sC gerbe} $\check \efF$ is given by a family of \sC dynamical systems
$\alpha_a : G \to {\bf aut}F_a$, $a \in \Delta$,
on which *-monomorphisms
$\jmath_{a'a} : F_a \to F_{a'}$, $a \subseteq a'$,
are defined in such a way that
\begin{equation}
\label{eq.CON2}
\jmath_{|c| c_1} \circ \jmath_{c_1 c_0} = \alpha_{|c|}(\delta_c) \circ \jmath_{|c| c_0}
\ \ \ , \ \ \
\jmath_{|b| b_0} \circ \alpha_{b_0}(g)  = \alpha_{|b|}(i_b(g)) \circ \jmath_{|b|b_0}
\ ,
\end{equation}
for any 
$b \in N_1(\Delta)$, $c \in N_2(\Delta)$, $g \in G$.
The first equality in (\ref{eq.CON2}) generalizes (\ref{eq.00}),
whilst the second one is analogous to (\ref{eq.GA00}). 
Moreover, again by (\ref{eq.CON2}), the fixed point family 
$\efR = (R,\jmath)_\Delta$,
$R_a := F_a^G$, $a \in \Delta$,
is a {\em net} of \sC algebras, that we interpret here as the observable net.
%
%
When $\delta \equiv 1$ the gerbe $\check \efG$ defines a group net bundle, written $\efG$, 
and $\check \efF$ is a $\efG$-{\em net} of \sC algebras.

To illustrate how gerbes can arise in quantum field theory we give the following result:
\begin{thm}
\label{thm.gauge3}
Let $M$ denote a globally hyperbolic (spin) manifold and $(G_0,\gamma)$ a $\bZ_2$-graded group. 
Given a symmetric, separating set $\lambda$ for $G_0$ and a morphism
\[
\chi : \pi_1(M) \to N^\lambda G_0 / G_0 \ ,
\]
there are a $G_0$-gerbe $\check \efG = (i,\delta)_\Delta$ and a $\check \efG$-\sC gerbe 
$\check \efF = (F,\jmath)_\Delta$,
defining a fixed-point \textbf{net} $\efR = (R,\jmath)_\Delta$ causally represented on a Hilbert net bundle.
\end{thm}

\begin{proof}
We take $\omega \in \Delta$, a path frame $P_\omega$, and define
\[
V^\da_b := \chi( p_{\omega \, \bo} * b * p_{\bl \, \omega} ) \in N^\lambda G_0 / G_0
\ \ , \ \
\forall b \in \Sigma_1(\Delta) \ .
\]
By homotopy invariance of $\chi$, the relations
\begin{equation}
\label{eq.gauge3.1}
V^\da_{\partial_0c} V^\da_{\partial_2c} = V^\da_{\partial_1c}
\ \ , \ \
c \in \Sigma_2(\Delta)
\ ,
\end{equation}
are fulfilled (see \cite[\S 2.2]{Ruz05}). For any $b \in \Sigma_1(b)$ we pick a $V_b \in N^\lambda G_0$
such that $V_b \,{\mathrm{mod}}G_0 = V^\da_b$, so (\ref{eq.gauge3.1}) implies
\begin{equation}
\label{eq.gauge3.2}
\delta_c := V_{\partial_0c} V_{\partial_2c} V_{\partial_1c}^{-1} \in G_0
\ \ , \ \
c \in \Sigma_2(\Delta)
\ .
\end{equation}
We define 
$\imath_b := \ad V_b$, $\forall b \in \Sigma_1(\Delta)$,
which yields, by (\ref{eq.gauge3.2}), the desired gerbe $\check \efG = (\imath,\delta)_\Delta$.
We now consider the normal $( \tau(G_0),\tau(\gamma) )$-net of von Neumann algebras $\efF_*$
of Lemma \ref{lem.gauge2}, extend $\tau$ to $N^\lambda G_0$ as in the proof of Theorem \ref{thm.gauge1} 
and set $G_* := \tau(G_0)$, so $\tau(NG_0)$ takes values in $N_\gamma G_*$.
Recalling the inclusion $N_1(\Delta) \subset \Sigma_1(\Delta)$ we define
\[
\jmath_{a'a} : F_a \to F_{a'}
\ \ , \ \
\jmath_{a'a}(t) := \{ \ad \tau(V_{(a,a')}) \}(t)
\ \ , \ \
\forall a \subseteq a' \in \Delta
\ , \
t \in F_a
\ .
\]
Moreover we consider the gauge actions
\[
\alpha_a : G_0 \to {\bf aut}F_a 
\ \ , \ \ 
\alpha_a(g) := \ad \tau(g)
\ \ , \ \
a \in \Delta \ , \ g \in G_0 \ .
\]
By (\ref{eq.gauge3.2}) we have
\[
\jmath_{|c| c_1} \circ \jmath_{c_1 c_0} = \alpha_{|c|}(\delta_c) \circ \jmath_{|c| c_0}
\ \ , \ \ 
\forall c \in N_2(\Delta) \ .
\]
Moreover, 
\[
\jmath_{|b| b_0} \circ \alpha_{b_0}(g) \ = \
\{ \ad \tau(V_b) \} \circ \ad \tau(g) \ = \ 
\ad \tau ( V_bg ) \ = \
\ad \tau (\imath_b(g) \, V_b) \ = \
\alpha_{|b|}(\imath_b(g)) \circ \jmath_{|b| b_0} \ ,
\]
so $\check \efF$ is a $\check \efG$-gerbe as desired, with fixed-point net $\efR = (R,\jmath)_\Delta$.
To construct a causal representation of $\efR$ we consider the Hilbert space $H_*$
of Lemma \ref{lem.gauge2} and the subspace $\ovl H_*$ of $G_0$-invariant elements.
Defining
$\ovl U_{a'a}v:= \tau(V_{(a,a')}) v$, $a \subseteq a'$, $v \in \ovl H_*$, we find, 
for any triple $c := \{ a \subseteq a' \subseteq a'' \in \Delta \}$,
\[
\ovl U_{a''a'} \ovl U_{a'a}v =
\tau(V_{(a',a'')} \, V_{(a,a')}) v =
\tau ( \delta_c \, V_{(a,a'')} ) v =
\tau (V_{(a,a'')} ) \, \tau(\delta'_c) v =
\tau (V_{(a,a'')} ) v =
\ovl U_{a''a}v \ ,
\]
where $\delta'_c := V_{(a,a'')}^{-1} \delta_c V_{(a,a'')} \in G_0$.
So 
$\ovl \efH = (\ovl H, \ovl U)_\Delta$, $\ovl H = \{ \ovl H_a \equiv \ovl H_* \}$, 
is a Hilbert net bundle.
We now define
\[
\ovl \pi_a : R_a \to B \ovl H_a
\ \ , \ \
\ovl \pi_a(t)v := tv
\ \ , \ \
\forall v \in \ovl H_a \subset H_*
\ , \
a \in \Delta 
\ .
\]
Since each $R_a$ is pointwise $G_0$-invariant we find that $tv \in \ovl H_a$, so the definition is well-posed.
Moreover
\[
\ad \ovl U_{a'a} \circ \ovl \pi_a(t) = 
\{ \ad \tau(V_{(a,a')}) \}(t) = 
\jmath_{a'a}(t) =
\ovl \pi_{a'} \circ  \jmath_{a'a}(t) \ ,
\]
for any $t \in R_a$, so the pair $(\ovl \pi,\ovl U)$ is a representation as desired.
Finally, to verify causality (in the sense of (\ref{def.norm})) 
we note that for any $o \perp a$, $t_1 \in R_a$, $t_0 \in R_o$,
\[
[ \pi_a(t_1) \, , \, \ad U_p \circ \pi_o(t_0) ]
\ = \
[ t_1 \, , \, \ad \tau(V_p)(t_0)  ]
\ \ , \ \
\forall p : o \to a 
\ ,
\]
and the last commutator is zero because 
$\{ \ad \tau(V)\}R_o \subseteq R_o$ for any $V \in N^\lambda G_0$ 
(recall (\ref{def.NG}) and that $\tau(N^\lambda G_0) \subseteq N_\gamma G_*$).
This suffices to prove causality, as $\efR$ has the same fibres of the fixed-point net of $\efF_*$,
which is causal.
\end{proof}

It is easily verified that $\chi$ has a lift to $N^\lambda G_0$ if, and only if,
there is a choice of $\{ V_b \}$ such that $V$ is a 1-cocycle (see \cite[Theorem 5.5]{Vas12}),
that is, if, and only if, $\check \efF$ can be arranged in such a way to be a net.
The point of the previous theorem is that we expect that the morphism $\chi$ is determined 
by the presheaf structure of $Z^{1,\bullet}(\efR)$, on the line of \cite[Theorem 5.3]{Vas12}.

\

\noindent {\it Acknowledgements.} The author would like to thank G. Ruzzi for fruitful discussions
and G. Morsella for precious comments on \S \ref{rem.qed}.

\appendix

\section{Symmetry and fields in globally hyperbolic manifolds.}
\label{app.sym}

\paragraph{The proof of the symmetry relations (\ref{eq.sym}).}
Before to make the necessary computations we need two Lemmata:

\begin{lem}
\label{lem.qft1}
For any $a \in \Delta$, cocycle $z_\rho \in \obj Z^1_c(\efR)_{< a}$ and paths $p,p' : e \to o$ in $\Delta^a$,
we have $z(p) = z(p')$.
\end{lem}

\begin{proof}
Let 
$\ovl p' : o \to e$
denote the opposite path.
Then $p * \ovl p' : o \to o$ is a loop in $\Delta^a$ and, since $\pi_1(\Delta) = \pi_1(a) = {\bf 0}$,
we have that $p * \ovl p'$ is homotopic to the trivial loop $b_o : o \to o$, $b_o := (o,o;o)$.
By \cite[Lemma 2.6-7]{Ruz05} we conclude that
$z(p)z(p')^* = z(p * \ovl p') = z(b_o) = 1$.
\end{proof}

The independence of $p$ allows us to write
$z_{oe} := z(p)$, $p : e \to o$;
note that $z_{eo} = z_{oe}^*$.

\begin{lem}
\label{lem.qft2}
The symmetry operator (\ref{eq.sym.z1}) is independent of the choice of 
$p_{oe}$, $p_{o'e}$, $o$ and $o'$,
for any $e < a$.
\end{lem}

\begin{proof}
The independence of $p_{oe}$, $p_{o'e}$ follows by the previous Lemma,
so it remains to verify the independence of $o,o'$. 
To this end, we have to verify that defining $\eps(z,w)_e$ with
$\omega , \omega' < a$, $\omega \perp \omega'$,
in place of $o,o'$ we obtain the same symmetry operator.
As an intermediate step, we consider a region $\bar{\omega}' < a$ such that 
$\bar{\omega}' \perp \omega , o$.
Then
\[
\begin{array}{lcl}
\eps(z,w)_e  & \stackrel{ Lemma \, \ref{lem.qft1} }{=} &
w_{\omega e}^* w_{o \omega}^* \sigma_o(z_{o'\bar{\omega}'}z_{\bar{\omega}' e})^* \cdot 
z_{o'\bar{\omega}'} z_{\bar{\omega}'e} \rho_e(w_{o \omega} w_{\omega e})
\\ & \stackrel{ o',\bar{\omega}' \perp o }{=} &
w_{\omega e}^* w_{o \omega}^* \sigma_o(z_{\bar{\omega}' e})^* \cdot 
z_{o'\bar{\omega}'}^* z_{o'\bar{\omega}'} z_{\bar{\omega}'e} \rho_e(w_{o \omega}) \rho_e(w_{\omega e})
\\ & \stackrel{ (\ref{eq.C1}) }{=} &
w_{\omega e}^*  \sigma_\omega(z_{\bar{\omega}' e})^* w_{o \omega}^* \cdot 
z_{o'\bar{\omega}'}^* z_{o'\bar{\omega}'} z_{\bar{\omega}'e} \rho_e(w_{o \omega}) \rho_e(w_{\omega e})
\\ & \stackrel{ (\ref{eq.C0}) }{=} &
w_{\omega e}^*  \sigma_\omega(z_{\bar{\omega}' e})^* \cdot 
\rho_{\bar{\omega}'}(w_{o \omega})^* z_{\bar{\omega}'e} \rho_e(w_{o \omega}) \rho_e(w_{\omega e})
\\ & \stackrel{ (\ref{eq.C1}) }{=} &
w_{\omega e}^*  \sigma_\omega(z_{\bar{\omega}' e})^* \cdot 
z_{\bar{\omega}'e} \rho_e( w_{o \omega}^* w_{o \omega}) \rho_e(w_{\omega e})
\\ & = &
w_{\omega e}^*  \sigma_\omega(z_{\bar{\omega}' e})^* \cdot
z_{\bar{\omega}'e} \rho_e(w_{\omega e}) \ .
\end{array}
\]
Let us write $\eps^{oo'}_e$ to denote the symmetry obtained by by using $o , o' < a$.
Then the previous computations say that 
\begin{equation}
\label{correct}
\eps^{oo'}_e \ = \ \eps^{\omega \bar{\omega}'}_e 
\ \ , \ \
\forall \bar{\omega}' < a \ : \ \bar{\omega}' \perp \omega,o \ .
\end{equation}
Let us now return to our $\omega' < a$ such that $\omega' \perp \omega$.
Then the argument for proving (\ref{correct}) applies, and we have
\[
\eps^{\omega \bar{\omega}'}_e \ = \ \eps^{\omega \omega'}_e \ ;
\]
but by (\ref{correct}) we also find $\eps^{oo'}_e = \eps^{\omega \omega'}_e$,
as desired.
\end{proof}

%
%
%
%
%
%
%
%

In the sequel, to be concise, we shall write 
$z \equiv z_\rho  ,  w_\sigma \equiv w  ,  v \equiv v_\tau  \in  Z^1_c(\efR)_{< a}$.
To prove (\ref{eq.sym}.1) we compute
\[
\begin{array}{lcl}
\eps(z,w)_\bo \cdot z(b) \rho_{\bl}(w(b)) & \stackrel{ (\ref{eq.C1}) }{=} &
w_{o \bo}^* \sigma_o(z_{o' \bo})^* \cdot z_{o' \bo} z(b) \rho_\bl(w_{o \bo} w(b))
\\ & \stackrel{ Lemma \, \ref{lem.qft2} }{=} &
w(b) \sigma_\bl(z(b)) \cdot 
\sigma_\bl(z(b))^* w(b)^* w_{o \bo}^* \sigma_o(z_{o' \bo})^* \cdot 
z_{o' \bl} \rho_\bl(w_{o \bl}) \\ & \stackrel{ Lemma \, \ref{lem.qft2} }{=} &
w(b) \sigma_\bl(z(b)) \cdot 
\sigma_\bl(z(b))^* w_{o \bl}^* \sigma_o(z_{o' \bo})^* \cdot 
z_{o' \bl} \rho_\bl(w_{o \bl}) \\ & \stackrel{ (\ref{eq.int0}) }{=} &
w(b) \sigma_\bl(z(b)) \cdot 
w_{o \bl}^* \sigma_o(z(b))^* \sigma_o(z_{o' \bo})^* \cdot 
z_{o' \bl} \rho_\bl(w_{o \bl}) \\ & = &
w(b) \sigma_\bl(z(b)) \cdot 
w_{o \bl}^* \sigma_o(z_{o' \bl})^* \cdot z_{o' \bl} \rho_\bl(w_{o \bl})
\\ & \stackrel{ Lemma \, \ref{lem.qft2} }{=} &
w(b) \sigma_\bl(z(b)) \cdot \eps(z,w)_\bl \ .
\end{array}
\]
To prove (\ref{eq.sym}.2) we compute
\[
\begin{array}{lcl}
s_e \sigma_e(t_e) \cdot \eps(z,w)_e & = &
s_e \sigma_e(t_e) \cdot w_{oe}^* \sigma_o(z_{o'e})^* \cdot z_{o'e} \rho_e(w_{oe}) 
\\ & \stackrel{ (\ref{eq.int}) }{=} &
\sigma'_e(t_e) s_e \cdot w_{oe}^* \sigma_o(z_{o'e})^* \cdot z_{o'e} \rho_e(w_{oe})  
\\ & \stackrel{ (\ref{eq.int00}) }{=} &
\sigma'_e(t_e) {w'}_{oe}^* \cdot s_o \sigma_o(z_{o'e})^* \cdot z_{o'e} \rho_e(w_{oe}) 
\\ & \stackrel{ (\ref{eq.int0}),(\ref{eq.int}) }{=} &
{w'}_{oe}^* \sigma'_o(t_e)  \cdot \sigma'_o(z_{o'e})^* s_o \cdot z_{o'e} \rho_e(w_{oe}) 
\\ & \stackrel{ o \perp o' , (\ref{eq.C0}) }{=} &
{w'}_{oe}^* \sigma'_o(t_e z_{o'e}^*) \cdot \rho_{o'}(s_o)  z_{o'e} \rho_e(w_{oe}) 
\\ & \stackrel{ (\ref{eq.int00}) , (\ref{eq.int0}) }{=} &
{w'}_{oe}^* \cdot \sigma'_o({z'}_{o'e})^*  \sigma'_o(t_{o'}) \cdot z_{o'e} \rho_e(s_o w_{oe}) 
\\ & \stackrel{ o \perp o' , (\ref{eq.C0}) }{=} &
{w'}_{oe}^* \cdot \sigma'_o({z'}_{o'e})^*  t_{o'} \cdot z_{o'e} \rho_e(s_o w_{oe}) 
\\ & \stackrel{ (\ref{eq.int00}) }{=} &
{w'}_{oe}^* \sigma'_o({z'}_{o'e})^*  \cdot t_{o'} z_{o'e} \rho_e(w'_{oe}) \cdot \rho_e(s_e) 
\\ & \stackrel{ (\ref{eq.int00}) }{=} &
{w'}_{oe}^* \sigma'_o({z'}_{o'e})^*  \cdot z'_{o'e} t_e \rho_e(w'_{oe}) \cdot \rho_e(s_e) 
\\ & \stackrel{ (\ref{eq.int}) }{=} &
{w'}_{oe}^* \sigma'_o({z'}_{o'e})^*  z'_{o'e} \rho'_e(w'_{oe}) \cdot t_e \rho_e(s_e) 
\\ & = &
\eps(z',w')_e \cdot t_e \rho_e(s_e) \ .
\end{array}
\]
To prove (\ref{eq.sym}.3) we compute
\[
\begin{array}{lcl}
\eps(z,w)_e \eps(w,z)  & = &
w_{oe}^* \sigma_o(z_{o'e})^* \cdot z_{o'e} \rho_e(w_{oe}) \cdot
z_{oe}^* \rho_o(w_{o'e})^* \cdot w_{o'e} \sigma_e(z_{oe})
\\ & \stackrel{ Lemma \ref{lem.qft2} }{=} &
w_{oe}^* \sigma_o(z_{o'e})^* \cdot z_{o'e} \cdot 
\rho_e(w_{oe}) z_{o'e}^* \rho_{o'}(w_{oe})^* \cdot w_{oe} \sigma_e(z_{o'e})
\\ & \stackrel{ (\ref{eq.C1}) }{=} &
w_{oe}^* \sigma_o(z_{o'e})^* \cdot z_{o'e} z_{o'e}^* \cdot 
\rho_{o'}(w_{oe}) \rho_{o'}(w_{oe})^* \cdot w_{oe} \sigma_e(z_{o'e})
\\ & \stackrel{ (\ref{eq.C1}) }{=} &
w_{oe}^* \sigma_o(z_{o'e})^* \cdot w_{oe} \sigma_e(z_{o'e})
\\ & \stackrel{ (\ref{eq.C1}) }{=} &
w_{oe}^* w_{oe} \sigma_e(z_{o'e})^* \sigma_e(z_{o'e})
\\ & = &
1 \ .
\end{array}
\]
The second of (\ref{eq.sym}.3) is trivial to verify.
Finally we prove (\ref{eq.sym}.4). To this end, we note that by \cite[Lemma B.5]{BR08}, 
given $e < a$ we can pick $o \perp o'$, $o,o' < a$, such that:
\begin{itemize}
\item $\ovl o \perp e$;
\item there is $\omega \supseteq e,o'$ such that $\omega \perp o$.
\end{itemize}
The second point implies, using Lemma \ref{lem.qft1}, that
$w_{o'e} = w(b_0) \in R_\omega$,
where $b_0 = (o',e;\omega) \in \Sigma_1(\Delta^a)$ can be regarded as a path from $e$ to $o'$.
So (\ref{eq.C0}) implies
\begin{equation}
\label{eq.app1}
\rho_{o'}(w_{o'e}) \in R_\omega
\ , \
\tau_o(w_{o'e}) = w_{o'e}
\ \ \ \Rightarrow \ \ \
\tau_o \circ \rho_{o'}(w_{o'e}) \ = \
\rho_{o'}(w_{o'e})              \ = \ 
\rho_{o'} \circ \tau_o(w_{o'e}) \ ,
\end{equation}
and we can compute
\[
\begin{array}{lcl}
\eps( (z \otimes w) ,v)_e  & = &
v_{oe}^* \tau_o( z_{o'e} \rho_e(w_{o'e}) )^* \cdot z_{o'e} \rho_e(w_{o'e}) \cdot
\{ \rho_e \circ \sigma_e \}(v_{oe})
\\ & \stackrel{ (\ref{eq.C1}) }{=} &
v_{oe}^* \tau_o( \rho_{o'}(w_{o'e}) z_{o'e} )^* \cdot z_{o'e} \rho_e(w_{o'e}) \cdot
\{ \rho_e \circ \sigma_e \}(v_{oe})
\\ & \stackrel{ (\ref{eq.app1}) }{=} &
v_{oe}^* \tau_o( z_{o'e})^* \cdot \{ \rho_{o'} \circ \tau_o \} (w_{o'e}) )^*  z_{o'e} \cdot
\rho_e ( w_{o'e} \sigma_e(v_{oe}) )
\\ & \stackrel{ (\ref{eq.C1}) }{=} &
v_{oe}^* \tau_o( z_{o'e})^* \cdot
z_{o'e} \rho_e (\tau_o(w_{o'e}))^*  \cdot  
\rho_e ( w_{o'e} \sigma_e(v_{oe}) )
\\ & = &
v_{oe}^* \tau_o( z_{o'e})^* z_{o'e} \rho_e(v_{oe}) \cdot
\rho_e ( v_{oe}^* \tau_o(w_{o'e})^* w_{o'e} \sigma_e(v_{oe}) )
\\ & = &
\eps(z,v)_e \cdot \rho_e( \eps(w,v)_e )
\ .
\end{array}
\]

\paragraph{Field nets on globally hyperbolic manifolds: the proof of Lemma \ref{lem.gauge2}.}
We apply the construction of \cite{DP02} to the quantum fields defined in \cite{Dim80,Dim82}
(in the following, we maintain the notation of \S \ref{sec.free} for 
$\lambda$, $L_\lambda^\pm$, $L_\lambda^{+,J}$, $DM$, $\mD^+$ and $\mD^-$).
To this end, following \cite{DP02} we introduce a mass function
\[
\mu : \lambda \to (0,\infty) \ \ : \ \ \inf \mu > 0 \ \ , \ \ \mu(\varrho) = \mu(\ovl{\varrho}) \ .
\]
Given a smooth Cauchy hypersurface $\Sigma \subset M$, we consider
the restriction operator and the forward normal derivative respectively,
\[
\rho_0 \, , \, \rho_1  \ : \ C^\infty(M) \to C^\infty(\Sigma) \ ,
\]
and, given $m>0$, the linear operator
\[
E_m : C^\infty_c(M) \to C^\infty(M)
\]
defined by the fundamental solution of the Klein-Gordon equation with mass $m$ 
and initial data in $\Sigma$ (see \cite[\S 2]{Dim80}).

Let $f \in \mD^+$ and $v \in L_\lambda^{+,J}$. We consider the decompositions
$f(x) = \{ f_\varrho(x) \in L_\varrho \}$, $v = \{ v_\varrho \in L_\varrho \}$,
and define the $C^\infty$, compactly supported functions
$f_{\varrho,v}(x) := ( f_\varrho(x) , v_\varrho )_\varrho \in \bR$, $x \in M$.
Note that $f_{\varrho,v}$ is linear in $v$, and
\[
( f(x) , v ) \ = \ \sum_\varrho f_{\varrho,v}(x) \ , \ \forall x \in M \ ,
\]
where \emph{the sum is finite} because $f \in \mD^+$.
Let now 
$C^\infty_{\mathrm{fin}}(M,L_\lambda^{+,J})$ 
denote the linear space of compactly supported,
$L_\lambda^{+,J}$-valued functions having a finite set of non-zero components on $\varrho \in \lambda$.
By Riesz duality in $L_\lambda^{+,J}$ and linearity of $\rho_0 , \rho_1 , E_m$ we can define the extensions
\[
\left\{
\begin{array}{ll}
E_\mu : \mD^+ \to C^\infty_{\mathrm{fin}}(M,L_\lambda^{+,J})
\ \ : \ \
( \{ E_\mu f \}(x) , v ) := \sum_{\varrho \in \lambda} \{ E_{\mu(\varrho)} f_{\varrho,v} \} (x)
\\ \\
\rho_\bullet^\lambda : C^\infty_{\mathrm{fin}}(M,L_\lambda^{+,J}) \to C^\infty_{\mathrm{fin}}(\Sigma,L_\lambda^{+,J})
\ \ : \ \
( \{ \rho_\bullet^\lambda f \}(x) , v ) := \sum_{\varrho \in \lambda} \{ \rho_\bullet f_{\varrho,v} \} (x)
\ \ , \ \
\bullet = 0,1
\ ,
\end{array}
\right.
\]
where $v$, $x$ vary in $L_\lambda^{+,J}$, $M$ respectively.

\smallskip

We now consider the Hilbert space $L^2(\Sigma,L_\lambda^+)$ and the associated bosonic Fock space, 
written $H_\Sigma^+$, carrying the representation of the canonical commutation relations
\[
[ a(h) , a^*(h') ] \ = \ \int_\Sigma (h,h')
\ \ , \ \
h , \tilde h \in L^2(\Sigma,L_\lambda^+)
\ .
\]
The argument of \cite[Theorem 2]{Dim80} shows that a representation of the canonical commutation relations is realized
on the Hilbert space $H_*^\phi := H_\Sigma^+$,
\[
[ \phi(f) , \phi(f') ] \ = \ -i \int_M (f , E_\mu f')
\ \ , \ \
\forall f,f' \in \mD^+
\ ,
\]
where
\[
\phi(f) := \theta(\rho_1^\lambda E_\mu f) - \pi(\rho_0^\lambda E_\mu f)
\ \ \ , \ \ \
\theta := 2^{-1/2}(a^*+a)
\ , \
\pi := i2^{-1/2}(a^*-a)
\ ,
\]
is well-defined because 
$\rho_1^\lambda E_\mu f , \rho_0^\lambda E_\mu f$ have compact support in $\Sigma$ 
(see the proof of \cite[Cor.1.3]{Dim80}).

\

We now consider the Dirac field, following \cite{Dim82}.
The strategy is the same of the scalar field, with the technical complications due to the spinor structures.
As a first step we consider the dual bundle $DM^* \to M$ of $DM$ and denote the set of $C^\infty$, 
compactly supported sections (cospinors) by $\mS^\infty_c(M,DM^*)$, so we have the pairing
$\left \langle \omega , s \right \rangle \in C^\infty_c(M,\bC)$,
$\omega \in \mS^\infty_c(M,DM^*)$, $s \in \mS^\infty_c(M,DM)$.
Again, we consider the restriction operator to the Cauchy hypersurface $\Sigma \subset M$ 
(endowed with the spin structure induced by $M$),
\begin{equation}
\label{eq.Dir1}
\rho : \mS^\infty(M,DM) \to \mS^\infty(\Sigma,D \Sigma) \ ,
\end{equation}
and the fundamental solution of the Dirac equation with mass $m > 0$,

\begin{equation}
\label{eq.Dir2}
S_m : \mS^\infty_c(M,DM) \to \mS^\infty(M,DM)
\end{equation}
(see \cite[\S II.A]{Dim82}). 
Moreover we consider the representation of vector fields as operators on the space of spinors,
defined, \emph{in local frames}, 
by contraction with the Clifford-Dirac $\gamma$-matrices (\cite[\S I.A]{Dim82}),
\begin{equation}
\label{eq.Dir3}
\mS^\infty(M,TM) \times \mS^\infty_c(M,DM) \to \mS^\infty_c(M,DM)
\ \ , \ \
\textbf{v} , s \mapsto \bar{\textbf{v}} s := ( \textbf{v}^h \gamma^i_{hk} s^k )_i \ ,
\end{equation}
and the Dirac conjugation, defined (in local frames) by applying to spinors the matrix $\beta \in \bSU(4)$ 
intertwining the $\gamma$-matrices with their adjoints (\cite[\S I.B]{Dim82}),
\begin{equation}
\label{eq.Dir4}
\mS^\infty_c(M,DM) \to \mS^\infty_c(M,DM^*) \ \ , \ \ s \mapsto s^+ := (s^h \beta_{hk})_k \ .
\end{equation}
%
%
If $\textbf{n}$ is the vector field forward normal to $\Sigma$, then we consider the integral
\begin{equation}
\label{eq.Dir5}
\int_\Sigma \left \langle s^+ , \bar{\textbf{n}} s' \right \rangle 
\ \ , \ \ 
s,s' \in \mS^\infty_c(\Sigma,D\Sigma)
\ ,
\end{equation}
that is a scalar product on $\mS^\infty_c(\Sigma,D\Sigma)$ (\cite[\S III.A]{Dim82}).

\smallskip

We now extend the coefficients considering the Hilbert bundles
$D_\lambda^-M \to M$ (\emph{generalized spinors})
and 
$D_\lambda^{-,*}M \to M$,
where $D^{-,*}_\lambda M  := L^{-,*}_\lambda \otimes DM^*$ is the tensor product
of the conjugate space $L^{-,*}_\lambda$ by the cospinor bundle (\emph{generalized cospinors}).
The restrictions to $\Sigma$ shall be denoted with analogous notations.
We have pairings 
\[
\mS^\infty_{c , {\mathrm{fin}} }(X,D_\lambda^{-,*}X) \times \mS^\infty_{c , {\mathrm{fin}} }(X,D_\lambda^-X)
\to
C^\infty_c(X,\bC)
\ \ , \ \
\omega , s \mapsto \left \langle \omega , s \right \rangle_\lambda^-
\ \ , \ \
X = M , \Sigma
\ ,
\]
where the symbol "fin" indicates that we are considering (co)spinors with finitely many components on 
$\varrho \in \lambda$.
We also have extensions of (\ref{eq.Dir1}) and (\ref{eq.Dir2}), respectively,
\[
\left\{
\begin{array}{ll}
\rho^\lambda : \mS^\infty(M,D_\lambda^-M) \to \mS^\infty(\Sigma,D_\lambda^-\Sigma)
\ \ , \ \
\rho^\lambda := \rho \otimes 1_- \ ,
\\ \\
S_\mu : \mD^- \to \mS^\infty(M,D_\lambda^-M)
\ \ , \ \
S_\mu s := \sum_\varrho^\oplus \{ S_{\mu(\varrho)} \otimes 1_- \} s_\varrho 
\end{array}
\right.
\]
(The sum of the previous expression is finite for any $s \in \mD^-$, so no convergence questions arise).

Moreover, (\ref{eq.Dir5}) yields a scalar product on $\mS^\infty_c(\Sigma,D_\lambda^-\Sigma)$ defined, 
on elementary tensors of the type
$s = s_1 \otimes s_2 \in \mS^\infty_c(\Sigma,D_\lambda^-\Sigma)$,
$s_1 \in \mS^\infty_c(\Sigma,D\Sigma)$,
$s_2 \in C^\infty_c(\Sigma,L^-_\lambda)$,
as
\[
\int_\Sigma \left\{ \left \langle s_1^+ , \bar{\textbf{n}} s'_1 \right \rangle  \,  (s_2 , s'_2) \right\}
\ \ , \ \ 
s,s' \in \mS^\infty_c(\Sigma,D_\lambda^-\Sigma)
\ .
\]
We denote the corresponding completion by $L^2(\Sigma,D_\lambda^-\Sigma)$.
Any generalized cospinor 
$\omega \in \mS^\infty_c(\Sigma,D_\lambda^{-,*}\Sigma)$ 
can be regarded as a linear functional on $L^2(\Sigma,D_\lambda^-\Sigma)$,
by extending the map
\[
s    \ \mapsto \    \int_\Sigma \left \langle \omega , s \right \rangle_\lambda^-
\ \ , \ \
s \in \mS^\infty_c(\Sigma,D_\lambda^-\Sigma)
\ ,
\]
to generic vectors in $L^2(\Sigma,D_\lambda^-\Sigma)$. 
So $\mS^\infty_c(\Sigma,D_\lambda^{-,*}\Sigma)$ embeds in the conjugate space
$L^{2,*}(\Sigma,D_\lambda^-\Sigma)$.

\smallskip

We now denote the fermionic Fock space of  $L^2(\Sigma,D_\lambda^-\Sigma)$ by $H_\Sigma^-$.
It carries a representation of the CARs
\[
X^*(h) X(\varphi) + X(\varphi)X^*(h) 
\ = \ 
\left \langle \varphi , h \right \rangle 
\ \ , \ \
h \in L^2(\Sigma,D_\lambda^-\Sigma) \ , \ \varphi \in L^{2,*}(\Sigma,D_\lambda^-\Sigma)
\ .
\]
We set $H_*^\psi := H_\Sigma^-$ and
\[
\psi(s) \ := \ -i X^*(\rho^\lambda S_\mu s)  \ \ , \ \  s \in \mD^- \ .
\]
These 
{\footnote{
In the reference \cite{Dim82} the analogue of the field $\psi$ is denoted by $\psi^+$,
whilst the notation $\psi$ is used for the Dirac field evaluated on cospinors.
}}
are well-defined operators in $BH_*^\psi$ since $\rho^\lambda S_\mu s$ has compact support 
for any $s \in \mD^-$ (see (\cite[III.B]{Dim82})). 
Recalling the definitions of $\rho^\lambda$ and $S_\mu$, 
and applying to them \cite[Prop.2.4(d)]{Dim82} we obtain, as desired, the CARs
\[
\psi^*(s) \psi(s') + \psi(s') \psi^*(s)
\ = \ 
-i \int_M \left \langle s^+ , S_\mu s' \right \rangle_\lambda^-
\ \ \ , \ \ \
s,s' \in \mD^- \ .
\]


\end{document}